\documentclass[a4paper,11pt]{article}
\pdfoutput=1 

\usepackage{jcappub} 

\usepackage[T1]{fontenc} 
\usepackage{amssymb}
\usepackage{wasysym}
\usepackage{mathrsfs}
\usepackage{txfonts}
\usepackage[usenames,dvipsnames]{xcolor}
\usepackage{graphicx}
\usepackage{epstopdf}
\usepackage{dcolumn}
\usepackage{bm}
\usepackage{bbm}
\usepackage{xcolor}
\usepackage{appendix}
\usepackage{amsmath,enumerate,bbm}
\usepackage[american]{babel}
\usepackage{ulem}

\newenvironment{enumeratenumeric}
               {\begin{enumerate}[1.] }{\end{enumerate}}
\newenvironment{proof}
               {\noindent\textbf{Proof\ }}{\hspace*{\fill}$\Box$\medskip}
\newtheorem{lemma}{Lemma}

\title{\boldmath Scalar, Vector and Tensor Harmonics on the Flat Compact
       Orientable Three-Manifolds}

\author[a]{Zhi-Peng Peng,}
\author[b,c]{Lee Lindblom,}
\author[a,d,1]{and Fan Zhang,\note{Corresponding author.}}

\affiliation[a]{Gravitational Wave and Cosmology Laboratory,
  Department of Astronomy, \\Beijing Normal University, Beijing 100875, China}
\affiliation[b]{Center for Astrophysics and Space Sciences,
  University of California at San Diego, \\La Jolla, CA 92093, USA}
\affiliation[c]{Center for Computational Mathematics,
  University of California at San Diego, \\La Jolla, CA 92093, USA}
\affiliation[d]{Department of Physics and Astronomy,
  West Virginia University, \\Morgantown, WV 26506, USA}
\emailAdd{zhipeng@mail.bnu.edu.cn}
\emailAdd{llindblom@ucsd.edu}
\emailAdd{fnzhang@bnu.edu.cn}

\abstract{  Observations suggest that our universe is spatially flat on the
  largest observable scales.  Exactly six different compact orientable
  three-dimensional manifolds admit flat metrics.  These six manifolds
  are therefore the most natural choices for building cosmological
  models based on the present observations.  This paper briefly
  describes these six manifolds and the harmonic basis functions
  previously developed for representing arbitrary scalar fields on
  them.  The principal focus of this paper is the development of new
  harmonics for representing arbitrary vector and second-rank tensor
  fields on these manifolds.  These new harmonics are designed to be
  useful tools for analyzing the dynamics of electromagnetic and
  gravitational fields on these spaces.
}

\begin{document}
\maketitle
\flushbottom

\section{\label{sec1}Introduction}
Current observations of the cosmic microwave radiation, together with
observations of lower-redshift distance indicators, show that the
large-scale average density of the universe differs by no more than
0.4\% (the present observational error estimate) from the critical
value that implies the geometry of our universe is spatially flat on
these largest observable scales~\cite{WMAP_9yr_Results}.  There are
eighteen different three-dimensional manifolds that admit flat
metrics~\cite{Feodoroff1885, Bieberbach1911, Novacki1934}, so these
are the most natural manifolds on which to construct realistic
cosmological models.  Ten of these flat manifolds are orientable,
while eight are non-orientable.  Spacetimes having non-orientable
spatial slices are not parallelizable (i.e. they do not admit
collections of smooth non-vanishing linearly independent vector
fields), and (consequently) such manifolds do not admit spinor
structures~\cite{Geroch1968, Geroch1970}.
Representations of fermions depend on the
  existence of these spinor structures, so unless (or perhaps until)
  new theoretical approaches for defining spinors are developed, the
  non-orientable three-manifolds appear to be less physically relevant
  than the orientable ones.  A purely emperical approach to the
  analysis of cosmological observations could justify the inclusion of
  the non-orientable cases anyway.  But given the theoretical spinor
  structure argument for excluding them, we have chosen to limit our
  analysis here to the orientable cases. Of the ten orientable flat
three-manifolds, six are compact while the remaining four are
non-compact.  Current observations do not allow us to see the entire
universe, so we have no way of knowing whether or not our universe is
spatially compact.  For computational convenience, we choose to limit
consideration here to models having compact spatial slices,
i.e. models with finite spatial volumes.

Each of the six compact orientable three-dimensional manifolds that
admits a flat metric can be obtained as a quotient
$\mathbb{E}^3/\Gamma$ of three-dimensional Euclidean space
$\mathbb{E}^3$ by an isometry group $\Gamma$ of symmetries of
$\mathbb{E}^3$.  The classification of these spaces has long been
known~\cite{Feodoroff1885, Bieberbach1911, Novacki1934}.  We use the
notation $E_1$, $E_2$, ..., $E_6$ to refer to these spaces.  The group
action that defines each of these spaces can be thought of as a
particular representation of $\mathbb{E}^3$ as a periodic lattice of
polytopes.  The individual flat compact manifolds can be thought of as
one of these polytopes with identifications between its faces.
Figure~1 of Ref.~\cite{riazuelo2004cosmic} illustrates these
polytope-with-identifications representations of these
manifolds.  The spaces $E_1$, $E_2$ and $E_3$ are based on rectangular
lattice representations of $\mathbb{E}^3$.  $E_1$ is the simple
three-torus, obtained by identifying opposite faces of a rectangular
solid.  $E_2$ and $E_3$ are obtained by identifying the opposite $x$
and $y$ faces of a rectangular sold, but twisting by $\pi$ before
identifying opposite $z$ faces for the half-turn space $E_2$, or by
$\pi/2$ for the quarter-turn space $E_3$.  $E_4$ and $E_5$ are
obtained from representations of $\mathbb{E}^3$ as a lattice of
hexagonal prisms.  These hexagonal prisms have six rectangular faces,
and two hexagonal faces.  The opposing rectangular faces of these
prisms are identified in $E_4$ and $E_5$, while the hexagonal faces
are twisted by $2\pi/3$ for the third-turn space $E_4$, or by $\pi/3$
for the sixth-turn space $E_5$.  The Hantzsche-Wendt space, $E_6$, is
based on a representation of $\mathbb{E}^3$ by a lattice of rhombic
dodecahedrons.  The space $E_6$ is formed by a particular
identification of the rhombic faces of one of these dodecahedrons.
Explicit descriptions of the symmetries used to create each of the
spaces, $E_1$, $E_2$, ..., $E_6$, are given in Sec.~\ref{sec3} as part
of our discussion of the scalar harmonics on these spaces.

We use the term harmonics in this paper to refer to the eigenfunctions
of the covariant Laplace operator.  These eigenfunctions form a
complete set of smooth functions on these compact orientable
manifolds, and can therefore be used as a basis for representing
arbitrary square integrable functions on them.  These harmonics can be
thought of as generalizations of the Fourier basis functions used to
construct representations of functions: $f(\mathbf{x}) =
\sum_\mathbf{k} f_\mathbf{\,k}\, e^{\,i\,\mathbf{k}\cdot\mathbf{x}}$.
Scalar harmonics have been developed for each of the six flat compact
orientable three-dimensional manifolds, and these harmonics have been
used to model the temperature variations in the cosmic microwave
background radiation that would be observed on these
spaces~\cite{riazuelo2004cosmic}.  Scalar harmonics, however, are not
adequate to model the full dynamics of the gravitational or the
electromagnetic fields.  For example, an analysis of
  the polarization properties of the cosmic microwave background
  radiation would require representations of the full vector structure
  of the electromagnetic field. A number
of studies have been carried out on the vector and tensor harmonics on
the (less physically relevant) manifolds with the topology of the
three-sphere, $\mathbb{S}^3$~\cite{sandberg1978tensor,
  jantzen1978tensor, rubin1984eigenvalues, rubin1985symmetric,
  lachieze2005laplacian, ben2016explicit, lindblom2017scalar}.  But
little attention has been paid to such matters on the less-familiar
three-manifolds that admit flat metrics.  There has been some work on
constructing second-rank tensor basis functions for $E_1$ in the
context of studying the effects of inhomogeneous initial conditions on
inflation~\cite{clough2018robustness}, or the possibility of
gravitational wave turbulence in the early
universe~\cite{clough2018difficulty}.  Here we significantly
generalize these studies by constructing complete vector harmonic and
second-rank tensor harmonic basis functions on all six flat compact
orientable three-manifolds.  To facilitate the analysis of the
dynamics of the gravitational and electromagnetic fields on these
manifolds, we have organized these new harmonics into subsets that
maximize the number of classes having vanishing divergence and trace.
These new harmonics were also constructed to satisfy nice
orthonormality conditions that make it easy to represent arbitrary
scalar and second-rank tensor fields on these manifolds.

The remainder of this paper is organized as follows.  Explicit
expressions for the symmetries used to construct each of the flat
compact orientable three-manifolds, $E_1$, $E_2$, ..., $E_6$, are
given in Sec.~\ref{sec3}, along with explicit expressions for the
(suitably re-normalized) scalar harmonics developed in
Ref.~\cite{riazuelo2004cosmic}. The analogous vector and
anti-symmetric second-rank tensor harmonics are developed in
Sec.~\ref{sec4}.  Two classes of these new vector harmonics are
divergence free, so they provide a natural way to represent the
transverse parts of dynamical electromagnetic fields.  Symmetric
second-rank tensor harmonics are developed in Sec.~\ref{sec5}. These
new tensor harmonics include five classes of trace-free harmonics, two
of which are also divergence-free.  These new tensor harmonics are
well suited therefore for representing the transverse-traceless parts
(i.e. the dynamical gravitational wave parts) of the gravitational
fields on these flat compact orientable three-manifolds.
Section~\ref{sec6} contains a brief summary and discussion of the new
results.  Several useful technical lemmas needed in the construction
of the new vector and tensor harmonics are given in an Appendix.

\section{\label{sec3}Scalar harmonics}

The scalar harmonics are defined here to be eigenfunctions of the
covariant Laplace operator:
\begin{equation}
  \nabla^a\nabla_a Y = -\kappa^2 Y.
  \label{e:ScalarLaplaceEigenvalueEq}
\end{equation}
We use the notation $Y[E_j]_{\,\mathbf{k}}$ to denote the harmonics on
the manifold $E_j$ (for $j=1,2,...,6$), where $\mathbf{k}= k_a=
(k_1,k_2,k_3)$ are parameters that identify a particular harmonic.  On
the three-torus, $E_1$, these harmonics are (up to normalizations)
just the Fourier basis functions:
\begin{equation}
  Y[E_1]_{\,\mathbf{k}}=\frac{e^{\,i\,\mathbf{k}\cdot\mathbf{x}}}{\sqrt{L_1L_2L_3}}.
  \label{e:ScalarE1Def}
\end{equation}
where $L_1$, $L_2$, and $L_3$, are the lengths of the three principal
axes of $E_1$.  The parameters $\mathbf{k}$ are chosen to ensure that
the harmonics have the appropriate periodicities on $\mathbb{E}^3$ to
make them smooth functions on $E_1$:
\begin{equation}
  \mathbf{k}=k_a = 2\pi\left(\frac{n_1}{L_1},\frac{n_2}{L_2},\frac{n_3}{L_3}\right).
\end{equation}
This requires the $n_a$ (for $a=1,2,3$) to be integers,
$n_a\in\mathbb{Z}$.  The eigenvalues for these
solutions to Eq.~(\ref{e:ScalarLaplaceEigenvalueEq}) are given by
\begin{equation}
  \kappa^2=k^ak_a = k_1^2 + k_2^2 +k_3^2.
  \label{e:ScalarEigenvalue}
\end{equation}
The normalization has been chosen in Eq.~(\ref{e:ScalarE1Def}) to
ensure that these harmonics satisfy the orthonormality conditions,
\begin{equation}
  \int_{-{L_1}/2}^{{L_1}/2}\int_{-{L_2}/2}^{{L_2}/2}\int_{-{L_3}/2}^{{L_3}/2}
  Y[E_1]_{\,\mathbf{k}}\, Y[E_1]_{\,\mathbf{k'}}^*\, dx\,dy\,dz=\delta_{{n_1}{n_1}'}
  \,\delta_{{n_2}{n_2}'}\,\delta_{{n_3}{n_3}'}.
  \label{e:ScalarOrthonormality}
\end{equation}

The harmonics for the remaining flat manifolds, $E_1$, ..., $E_6$,
were derived in Ref.~\cite{riazuelo2004cosmic} by using the fact that
each scalar harmonic $Y_{\text{\textbf{k}}}$ of $\mathbb{E}^3 /
\Gamma$ lifts to a $\Gamma-$periodic harmonic $Y_{\text{\textbf{k}}}$
of $\mathbb{E}^3$, i.e. to a scalar harmonic of $\mathbb{E}^3$ that is
invariant under the action of the isometry group $\Gamma$.  Finding
the scalar harmonics of the flat space $\mathbb{E}^3 / \Gamma$ is
equivalent, therefore, to finding the $\Gamma${\text{-}}periodic
scalar harmonics of $\mathbb{E}^3$.  Each element of an isometry group
$\Gamma$ of Euclidean space $\mathbb{E}^3$, can be written as a
rotation/reflection $\mathbf{M}$ followed by a translation
$\mathbf{T}$: i.e. these isometries map points
$\mathbf{x}\in\mathbb{E}^3$ to the points
$\mathbf{x'}=\mathbf{M}\cdot\mathbf{x}+\mathbf{T}$, or equivalently in
component notation $x'^a = M^a{}_b\,x^b+T^a$.  These transformations
are isometries so they preserve the forms of the metric
$\mathbf{g}=g_{ab}=\mathrm{diag}(1,1,1)$ and the inverse metric
$\mathbf{g}^{-1}=g^{ab}=\mathrm{diag}(1,1,1)$.  This implies that
$g_{ab}=M^c{}_a\,M^d{}_b\,g_{cd}$ and $g^{ab} =
M^a{}_c\,M^b{}_d\,g^{cd}$, and consequently that
$\kappa^2=k^ak_a=\mathbf{g}^{-1}(\mathbf{k},\mathbf{k})$ and
$\kappa^2= k_aM^a{}_c\, k_b\, M^b{}_d \,g^{cd}
=\mathbf{g}^{-1}(\mathbf{k}\mathbf{M},\mathbf{k\mathbf{M}})$ for these
isometries.

The isometry group $\Gamma$ of the half-turn space, $E_2$, consists of
pure translations by $L_1$, $L_2$ or $L_3$ in the principal
directions, plus a compound rotation-translation defined by
\begin{equation}
\begin{array}{rclcrcl}
  \mathbf{M}{[ E_2]} &=
&\left( \begin{array}{ccc} - 1 & 0 & 0\\ 0 & - 1 & 0\\ 0 & 0 &
    1\end{array} \right), &\qquad\qquad\qquad
    & \mathbf{T}[E_2]&=&\left(\begin{array}{c}
    0\\ 0\\ \tfrac{1}{2}L_3\end{array}\right).
\end{array}
\end{equation}
The scalar harmonics invariant under these transformations are given
by~\cite{riazuelo2004cosmic},
\begin{eqnarray}
  Y{[ E_2]}_{\textbf{\text{k}}} & = & \tfrac{1}{\sqrt{2}} \left[ Y{[
        E_1]}_{\textbf{\text{k}}} + ( - 1)^{n_3} Y{[
        E_1]}_{\textbf{\text{k}} \mathbf{M}{[ E_2]}} \right] \hspace{1em}
\notag \\ &&\,\,\,\, \quad \quad \quad \quad \quad
\text{for} \hspace{1em} \left( n_1 \in \mathbb{Z}^+, n_2, n_3 \in
  \mathbb{Z}\right) \hspace{1em} \text{or} \hspace{1em} \left( n_1 = 0, n_2 \in
  \mathbb{Z}^+, n_3 \in \mathbb{Z}\right), \label{e:ScalarE2Def1}\\
  Y{[ E_2]}_{(0, 0, k_3)} & = &
  Y{[ E_1]}_{(0, 0, k_3)} \hspace{1em} \text{for}
  \hspace{1em}  \left(n_1 = n_2 = 0, n_3 \in 2\mathbb{Z}\right),\label{e:ScalarE2Def2}
\end{eqnarray}
where the $Y[E_1]_\mathbf{k}$ are the basic $E_1$ harmonics given in
Eq.~(\ref{e:ScalarE1Def}), and $\textbf{\text{k}} = ( k_1, k_2, k_3) =
2 \pi\left( \frac{n_1}{L_1},\frac{n_2}{L_2}, \frac{n_3}{L_3}\right).$
These harmonics have the same eigenvalues as the $E_1$ harmonics,
Eq.~(\ref{e:ScalarEigenvalue}), and satisfy orthonormality
conditions analogous to Eq.~(\ref{e:ScalarOrthonormality}).

The isometry group of the quarter-turn space, $E_3$, consists of pure
translations by $L_1$, $L_2$ or $L_3$ in the principal directions,
plus a compound rotation-translation defined by
\begin{equation}
\begin{array}{rclcrcl}
  \mathbf{M}{[ E_3]} &=
&\left( \begin{array}{ccc} 0 & - 1 & 0\\ 1 & 0 & 0\\ 0 & 0 &
    1\end{array} \right), &\qquad\qquad\qquad
    & \mathbf{T}[E_3]&=&\left(\begin{array}{c}
    0\\ 0\\ \tfrac{1}{4}L_3\end{array}\right).
\end{array}
\end{equation}
These symmetries imply that $L_1=L_2$ in the $E_3$ case.  The scalar
harmonics invariant under these transformations are given
by~\cite{riazuelo2004cosmic},
\begin{eqnarray}
    Y{[ E_3]}_{\textbf{\text{k}}} &=&  \tfrac{1}{2}
  \left[ Y{[E_1]}_{\textbf{\text{k}}}
    + {i}^{\,n_3} Y{[E_1]}_{\textbf{\text{k}} M{[ E_3]}}
    + {i}^{\,2 n_3} Y{[E_1]}_{\textbf{\text{k}} M{[ E_3]^2}}
    + {i}^{\,3 n_3} Y{[E_1]}_{\textbf{\text{k}} M{[ E_3]^3}} \right] \hspace{1em}
\notag \\ &&\, \,\,\, \quad \quad \quad \quad \quad
\text{for} \hspace{1em} \left(n_1 \in \mathbb{Z}^+, n_2 \in \mathbb{Z}^+
  \cup \{ 0 \}, n_3 \in \mathbb{Z}^{}\right), \\
  \!\!\!\!\!\!\!\!\!\!\!
  Y{[ E_3]}_{(0, 0, k_3)}  &=&  Y{[ E_1]}_{(0, 0, k_3)} \hspace{1em} \text{for}
  \hspace{1em}  \left(n_1 = n_2 = 0, n_3 \in 4\mathbb{Z}\right),
\end{eqnarray}
where the $Y[E_1]_\mathbf{k}$ are the basic $E_1$ harmonics given in
Eq.~(\ref{e:ScalarE1Def}), and $\textbf{\text{k}} = ( k_1, k_2, k_3) =
2 \pi\left( \frac{n_1}{L_1},\frac{n_2}{L_2}, \frac{n_3}{L_3}\right).$
These harmonics have the same eigenvalues as the $E_1$ harmonics,
Eq.~(\ref{e:ScalarEigenvalue}), and satisfy orthonormality
conditions analogous to Eq.~(\ref{e:ScalarOrthonormality}).

The spaces $E_4$ and $E_5$ are constructed from a hexagonal prism
lattice representation of $\mathbb{E}^3$ that is generated by the four
pure-translation symmetries
\begin{equation}
  \begin{array}{lclclcl}
    \mathbf{T}_1=\left(\begin{array}{c}L_1 \\ 0 \\ 0\end{array}\right),&\qquad
      &\mathbf{T}_2=\left(\begin{array}{c}-\tfrac{1}{2}L_1
        \\ \tfrac{\sqrt{3}}{2}L_2 \\ 0\end{array}\right),&\qquad
      &\mathbf{T}_3=\left(\begin{array}{c}-\tfrac{1}{2}L_1
        \\ -\tfrac{\sqrt{3}}{2}L_2 \\ 0\end{array}\right),&\qquad
        &\mathbf{T}_4=\left(\begin{array}{c}0 \\ 0
          \\ L_3\end{array}\right).
          \label{e:HexagonalT}
  \end{array}
\end{equation}
The isometry group of the third-turn space, $E_4$, consists of the
pure translations given in Eq.~(\ref{e:HexagonalT}) plus a compound
rotation-translation defined by
\begin{equation}
\begin{array}{rclcrcl}
  \mathbf{M}{[ E_4]} &=
  &\left( \begin{array}{ccc} -\tfrac{1}{2} & - \tfrac{\sqrt{3}}{2} & 0\\
    \tfrac{\sqrt{3}}{2}
    & -\tfrac{1}{2} & 0\\ 0 & 0 &
    1\end{array} \right), &\qquad\qquad\qquad
    & \mathbf{T}[E_4]&=&\left(\begin{array}{c}
    0\\ 0\\ \tfrac{1}{3}L_3\end{array}\right).
\end{array}
\end{equation}
These symmetries imply that $L_1=L_2$ in the $E_4$ case.  The scalar
harmonics invariant under these transformations are given
by~\cite{riazuelo2004cosmic},
\begin{eqnarray}
  Y{[ E_4]}_{\textbf{\text{k}}} & = & \tfrac{1}{\sqrt{3}} \left[ Y{[
        E_1]}_{\textbf{\textbf{\text{k}}}} + \omega^{\,2n_3} Y{[
        E_1]}_{\textbf{\text{k}}\mathbf{M}{[ E_4]}} + \omega^{\,4 n_3}
      Y{[ E_1]}_{\text{\textbf{k}} \mathbf{M}{[ E_4]^2}}
      \right] \hspace{1em}
      \notag \\ &&\,\,\,\, \quad \quad \quad \quad \quad
      \text{for}
  \hspace{1em} \left(n_1 \in \mathbb{Z}^+, n_2 \in \mathbb{Z}^+ \cup \{ 0 \}, n_3
  \in \mathbb{Z}\right), \\
  Y{[ E_4]}_{(0, 0, k_3)} & = & Y{[ E_1]}_{(0, 0, k_3)} \hspace{1em} \text{for}
  \hspace{1em}  \left(n_1 = n_2 = 0, n_3 \in 3 \mathbb{Z}\right),
\end{eqnarray}
where the $Y[E_1]_\mathbf{k}$ are the basic $E_1$ harmonics given in
Eq.~(\ref{e:ScalarE1Def}), and $\omega = e^{\,i\, \pi / 3}$.  To
preserve the hexagonal translation symmetry in this case we must also
take $\text{\textbf{k}} = ( k_1, k_2, k_3) = 2 \pi\left( -
\tfrac{n_2}{L_1}, \tfrac{2 n_1 - n_2}{\sqrt{3}L_2}, \tfrac{n_3}{L_3}
\right)$.  These harmonics have the same eigenvalues as the $E_1$
harmonics, Eq.~(\ref{e:ScalarEigenvalue}), and satisfy orthonormality
conditions analogous to Eq.~(\ref{e:ScalarOrthonormality}).

The isometry group of the sixth-turn space, $E_5$, consists of the
pure hexagonal lattice translations given in Eq.~(\ref{e:HexagonalT}),
plus a compound rotation-translation defined by
\begin{equation}
\begin{array}{rclcrcl}
  \mathbf{M}{[ E_5]} &=
  &\left( \begin{array}{ccc} \tfrac{1}{2} & - \tfrac{\sqrt{3}}{2} & 0\\
    \tfrac{\sqrt{3}}{2}
    & \tfrac{1}{2} & 0\\ 0 & 0 &
    1\end{array} \right), &\qquad\qquad\qquad
    & \mathbf{T}[E_5]&=&\left(\begin{array}{c}
    0\\ 0\\ \tfrac{1}{6}L_3\end{array}\right).
\end{array}
\end{equation}
These symmetries imply that $L_1=L_2$ in the $E_5$ case.  The scalar
harmonics invariant under these transformations are given
by~\cite{riazuelo2004cosmic},
\begin{eqnarray}
  Y{[ E_5]}_{\text{\textbf{k}}} & = & \tfrac{1}{\sqrt{6}} \Big[
    Y{[ E_1]}_{\textbf{\text{k}}}
    + \omega^{n_3} Y[E_1]_{\textbf{\text{k}}\mathbf{M}{[E_5]}}
    + \omega^{2 n_3} Y[E_1]_{\textbf{\text{k} }\mathbf{M}{[E_5]^2}}
    + \omega^{3 n_3} Y[E_1]_{\textbf{\text{k} }\mathbf{M}{[E_5]^3}}
    \notag \\ &&
    + \omega^{4 n_3} Y[E_1]_{\textbf{\text{k} }\mathbf{M}{[E_5]^4}}
    + \omega^{5 n_3} Y[E_1]_{\textbf{\text{k}}\mathbf{M}{[E_5]^5}}
    \Big] \hspace{1em} \nonumber  \\
  &&\quad\quad\quad\,\,\,\,\qquad\text{for}\hspace{1em}
   \left(n_1 \in \mathbb{Z}^+, n_2 \in \mathbb{Z}^+ \cup \{ 0 \}, n_2 <
  n_1, n_3 \in \mathbb{Z}^{}\right), \\
  Y{[ E_5]}_{(0, 0, k_3)} & = & Y{[ E_1]}_{(0, 0, k_3)} \hspace{1em} \text{for}
  \hspace{1em} \left(n_1 = n_2 = 0, n_3 \in 6 \mathbb{Z}\right),
\end{eqnarray}
where the $Y[E_1]_\mathbf{k}$ are the basic $E_1$ harmonics given in
Eq.~(\ref{e:ScalarE1Def}), and $\omega = e^{\,i\, \pi / 3}$.  To
preserve the hexagonal translation symmetry in this case we must also
take $\text{\textbf{k}} = ( k_1, k_2, k_3) = 2 \pi\left( -
\tfrac{n_2}{L_1}, \tfrac{2 n_1 - n_2}{\sqrt{3}L_2}, \tfrac{n_3}{L_3}
\right)$.  These harmonics have the same eigenvalues as the $E_1$
harmonics, Eq.~(\ref{e:ScalarEigenvalue}), and satisfy orthonormality
conditions analogous to Eq.~(\ref{e:ScalarOrthonormality}).

The isometry group of the Hantzsche-Wendt space, $E_6$, consists of
pure translations by $L_1$, $L_2$ or $L_3$ in the principal
directions, plus compound rotation-translations defined by
\begin{equation}
\begin{array}{rclcrcl}
  \mathbf{M}_1{[ E_6]} &=
  &\left( \begin{array}{ccc} 1& 0 & 0\\ 0& -1 & 0\\ 0 & 0 & -1\end{array} \right),
    &\qquad\qquad\qquad
    & \mathbf{T}_1[E_6]&=&\left(\begin{array}{c}
      \tfrac{1}{2}L_1\\ \tfrac{1}{2}L_2\\ 0\end{array}\right),\\
  \mathbf{M}_2{[ E_6]} &=
  &\left( \begin{array}{ccc} -1& 0 & 0\\ 0 & 1 & 0\\ 0 & 0 & -1\end{array} \right),
    &\qquad\qquad\qquad
    & \mathbf{T}_2[E_6]&=&\left(\begin{array}{c}
      0\\ \tfrac{1}{2}L_2\\ \tfrac{1}{2}L_3\end{array}\right),\\
  \mathbf{M}_3{[ E_6]} &=
  &\left( \begin{array}{ccc} -1& 0 & 0\\ 0 & -1 & 0\\ 0 & 0 & 1\end{array} \right),
    &\qquad\qquad\qquad
    & \mathbf{T}_3[E_6]&=&\left(\begin{array}{c}
      \tfrac{1}{2}L_1\\ 0 \\ \tfrac{1}{2}L_3\end{array}\right).
\end{array}
\end{equation}
The scalar harmonics invariant under these transformations are given
by~\cite{riazuelo2004cosmic},
\begin{eqnarray}
  Y{[ E_6]}_{\text{\textbf{k}}} & = &
  \tfrac{1}{2} \left[ Y{[E_1]}_{\text{\textbf{k}}}
    + ( - 1)^{n_1 - n_2} Y{[E_1]}_{\text{\textbf{k}} \mathbf{M}_1{[E_6]}}
    + ( - 1)^{n_2 - n_3} Y{[E_1]}_{\text{\textbf{k}} \mathbf{M}_2{[E_6]}}
    + ( - 1)^{n_3 - n_1} Y{[E_1]}_{\text{\textbf{k}} \mathbf{M}_3{[E_6]}}
    \right] \nonumber  \\ &&\qquad \text{for}\hspace{1em}\left( n_1, n_2
  \in \mathbb{Z}^+, n_3 \in \mathbb{Z} \right) \hspace{1em} \text{or}
  \hspace{1em} \left( n_1 = 0, n_2, n_3 \in \mathbb{Z}^+ \right)
  \hspace{1em}
  \nonumber  \\ &&\qquad
  \text{or} \hspace{1em} \left( n_2 = 0, n_1, n_3 \in
  \mathbb{Z}^+ \right), \\
  Y{[ E_6]}_{(k_1, 0, 0)} & = & \frac{1}{\sqrt{2}} \left[ Y{[ E_1]}_{(k_1, 0, 0)} +
    Y{[ E_1]}_{(- k_1, 0, 0)}\right] \hspace{1em}
  \text{for} \hspace{1em} \left( n_1 \in 2 \mathbb{Z}^+, n_2=n_3=0\right), \\
  Y{[ E_6]}_{(0, k_2, 0)} & = & \frac{1}{\sqrt{2}} \left[ Y{[ E_1]}_{(0, k_2, 0)} +
    Y{[ E_1]}_{(0, - k_2, 0)}\right] \hspace{1em}
  \text{for} \hspace{1em} \left( n_2 \in 2\mathbb{Z}^+, n_1=n_3=0\right), \\
  Y{[ E_6]}_{(0, 0, k_3)} & = & \frac{1}{\sqrt{2}} \left[ Y{[ E_1]}_{(0, 0, k_3)} +
    Y{[ E_1]}_{(0, 0, - k_3)}\right] \hspace{1em}
  \text{for} \hspace{1em} \left( n_3 \in 2 \mathbb{Z}^+, n_1=n_2=0\right),
\end{eqnarray}
where the $Y[E_1]_\mathbf{k}$ are the basic $E_1$ harmonics given in
Eq.~(\ref{e:ScalarE1Def}), and $\textbf{\text{k}} = ( k_1, k_2, k_3) =
2 \pi\left( \frac{n_1}{L_1},\frac{n_2}{L_2}, \frac{n_3}{L_3}\right).$
These harmonics have the same eigenvalues as the $E_1$ harmonics,
Eq.~(\ref{e:ScalarEigenvalue}), and satisfy orthonormality
conditions analogous to Eq.~(\ref{e:ScalarOrthonormality}).

\section{\label{sec4}Vector harmonics}

This section constructs vector harmonics on the six flat compact
orientable three-manifolds.  While these harmonics are not unique and
can be chosen in a variety of different way, our goal here is to
choose harmonics having three useful properties: First, the vector
harmonics constructed here will be eigenfunctions of the covariant
Laplace operator:
\begin{equation}
  \nabla^b\nabla_b Y^a = -\kappa^2 Y^a.
  \label{e:VectorLaplaceEigenvalueEq}
\end{equation}
This ensures that these harmonics will form a complete basis for the
(square integrable) vector fields on these manifolds.  We use the
notation $Y[E_j]^a_{(A)\,\mathbf{k}}$ (for $A=0,1,2$) to denote the
three linearly independent classes of vector harmonics that satisfy
Eq.~(\ref{e:VectorLaplaceEigenvalueEq}) on the space $E_j$ (for
$j=1,2,...,6$).  Second, the vector harmonics constructed here will
satisfy nice orthonormality conditions,
i.e. $Y[E_j]^a_{(A)\,\mathbf{k}}$ and $Y[E_j]^a_{(A')\,\mathbf{k'}}$ will
be orthogonal under the standard $L_2$ inner product unless
$\mathbf{k}=\mathbf{k'}$ and $A=A'$.  And third, the vector harmonics
constructed here in classes $A=1$ and $A=2$ will be divergence free.

It is easy to construct one nice class of vector harmonics on all the
flat compact orientable three-manifolds.  These harmonics, which we
call the class $A=0$ harmonics, are defined as the gradients of the
scalar harmonics on each manifold:
\begin{equation}
  Y[E_j]^a_{(0)\,\mathbf{k}}=\kappa^{-1}\nabla^a Y[E_j]_{\mathbf{k}},
  \label{e:VectorHarmonic0Def}
\end{equation}
where $\nabla^a=g^{ab}\nabla_b$ and $\kappa^2$ is the corresponding
eigenvalue of the covariant scalar Laplace operator,
Eq.~(\ref{e:ScalarEigenvalue}).  Since the scalar harmonics
$Y[E_j]_{\mathbf{k}}$ are smooth eigenfunctions of the Laplace
operator, their gradients are automatically smooth eigenfunctions
having the same eigenvalues on these flat manifolds.  The
normalization factor in Eq.~(\ref{e:VectorHarmonic0Def}) is chosen to
ensure that these harmonics satisfy the nice orthonormality conditions
\begin{equation}
  \int_{- L _1/ 2}^{L_1 / 2} \int_{- L_2 / 2}^{L _2/ 2} \int_{- L_3 / 2}^{L_3 / 2} g_{ab}Y{[
  E_j]}^a_{(0)\,\mathbf{k}} Y{[E_j]}^{b*}_{(0)\,\mathbf{k}'} \,d x\, d y\, d z =
  \delta_{{n_1} {n_1}'}\,\delta_{{n_2} {n_2}'}\,\delta_{{n_3} {n_3}'}\,.
\end{equation}
The divergences of these class $A=0$ harmonics are given by
\begin{equation}
  \nabla_aY{[E_j]}^a_{(0)\,\mathbf{k}}=-\kappa Y[E_j]_{\mathbf{k}}.
  \label{e:VectorHarmonicDivergence0}
\end{equation}

The construction of the class $A=1$ and $A=2$ vector harmonics is less
straightforward.  We have developed three different approaches that
produce vector harmonics satisfying the three useful properties listed
above.  The first approach produces the simplest expressions for the
vector harmonics, but this approach only works in the space $E_1$.  A
somewhat more general approach can be used to derive fairly simple
expressions for these vector harmonics in the spaces $E_1,...,E_5$, but
not in $E_6$.  We have also developed an even more general approach
capable of constructing vector harmonics in all the $E_j$ spaces, but
the resulting expressions produced in this way are quite complicated.
Here we report the results of this general approach only for the
otherwise intractable $E_6$ case.

The first approach to constructing the needed vector harmonics is
based on the fact that any covariently constant vector field $c^a$ is
invariant under the pure translation symmetry group of the
three-torus, $E_1$.  Any vector field of the form
$c^aY[E_1]_\mathbf{k}$ is an eigenfunction of the covariant Laplace
operator, and therefore a candidate vector harmonic.  The choices for
three linearly independent vectors $c^a$ will define the three classes
of vector harmonics for this case.  The class $A=0$ vector harmonics
defined in Eq.~(\ref{e:VectorHarmonic0Def}) for the space $E_1$ have
this form, $Y[E_1]^a_{(0)\,\mathbf{k}}=i\,\kappa^{-1}g^{ab}k_b
Y[E_1]_{\mathbf{k}}$ with $c^a=i\,\kappa^{-1}g^{ab}k_b$.  All that is
needed to complete this simple approach is to choose two additional
constant vectors to define the class $A=1$ and $A=2$ harmonics.  Any
unit vectors orthogonal to $k_a$ will do.  One choice is
$\ell_a=(-k_2-k_3,k_1-k_3,k_1+k_2)$ and
$m_a=g_{ab}\epsilon^{bcd}k_c\ell_d$, where $\epsilon^{abc}$ is the
covariantly constant, $\nabla_a\epsilon^{bcd}=0$, totally
antisymmetric tensor volume element.  For convenience, these vectors
can be normalized by setting $\hat k_a = \kappa^{-1}k_a$, $\hat\ell_a
= \lambda^{-1}\ell_a$, and $\hat m_a = \kappa^{-1}\lambda^{-1}m_a$,
where $\lambda^2=\ell^a\ell_a=(k_2+k_3)^2+(k_1-k_3)^2+(k_1+k_2)^2$.
Vector harmonics defined in terms of the orthonormal constant vectors
$\hat k^a=g^{ab}\hat k_b$, $\hat \ell^a=g^{ab}\hat\ell_b$, and $\hat
m^a=g^{ab}\hat m_b$ are given by
\begin{eqnarray}
  Y[E_1]^a_{(0)\,\mathbf{k}}&=&i\,\hat k^aY[E_1]_{\mathbf{k}},
  \label{e:VectorHarmonicE1Def0}\\
  Y[E_1]^a_{(1)\,\mathbf{k}}&=&i\,\hat\ell^aY[E_1]_{\mathbf{k}},
  \label{e:VectorHarmonicE1Def1}\\
  Y[E_1]^a_{(2)\,\mathbf{k}}&=&i\,\hat m^aY[E_1]_{\mathbf{k}},
  \label{e:VectorHarmonicE1Def2}
\end{eqnarray}
Special forms for $\hat k_a$, $\hat\ell_a$, and $\hat m_a$ are needed
when $\kappa=0$ or $\lambda=0$.  In the $\kappa=0$ case we can define
$\hat k_a=(1,0,0)$, $\hat \ell_a=(0,1,0)$ and $\hat m_a=(0,0,1)$.  In
the $\lambda=0$ but $\kappa\neq 0$ case $\hat
k_a=\tfrac{1}{\sqrt{3}}(1,-1,1)$ so we can simply define $\hat
\ell_a=\tfrac{1}{\sqrt{2}}(1,1,0)$ and $\hat
m_a=\tfrac{1}{\sqrt{6}}(-1,1,2)$.  All these vector harmonics defined
in Eqs.~(\ref{e:VectorHarmonicE1Def0})--(\ref{e:VectorHarmonicE1Def2})
satisfy each of the useful properties listed above.  They are
eigenfunctions of the covariant Laplace operator with eigenvalue
$-\kappa^2$, and satisfy the following orthonormality conditions,
\begin{equation}
  \int_{- L _1/ 2}^{L_1 / 2} \int_{- L_2 / 2}^{L _2/ 2} \int_{- L_3 / 2}^{L_3 / 2} g_{ab}Y{[
  E_1]}^a_{(A)\,\mathbf{k}} Y{[ E_1]}^{b*}_{(B)\,\mathbf{k}'} \,d x\, d y\, d z =
  \delta_{AB}\,\,\delta_{{n_1} {n_1}'}\,\delta_{{n_2} {n_2}'}\,\delta_{{n_3} {n_3}'}\,.
  \label{e:VectorOrthonormalityE1}
\end{equation}
The divergences of these vector harmonics on the space $E_1$
satisfy,
\begin{eqnarray}
  \nabla_aY{[E_1]}^a_{(0)\,\mathbf{k}}&=&-\kappa Y[E_1]_{\mathbf{k}},
  \label{e:VectorDivergenceE1Class0}\\
  \nabla_aY{[E_1]}^a_{(1)\,\mathbf{k}}&=&
  \nabla_aY{[E_1]}^a_{(2)\,\mathbf{k}}=0.\label{e:VectorDivergenceE1Class2}
\end{eqnarray}
The divergences of these class $A=1$ and $A=2$ vector harmonics vanish
whenever $\hat\ell_a$ and $\hat m_a$ are chosen to be orthogonal to
$\hat k_a$.

The second approach to constructing vector harmonics uses the fact
that the unit vector along the $z$-axis, $\hat z^a$, is the only unit
vector field invariant under all the symmetry groups of the manifolds
$E_1,...,E_5$.  Therefore $\hat z^aY[E_j]_\mathbf{k}$ (for
$j=1,...,5$) is an eigenfunction of the covariant Laplace operator,
and can be used in the construction of the vector harmonics on these
spaces in much the same way the constant vectors $c^a$ were used in
the first approach.  We note that $\hat
z^a\nabla_aY[E_j]_\mathbf{k}=i\,k_3 Y[E_j]_\mathbf{k}$ in all of these
cases. The vector harmonics for these manifolds can therefore be
taken to be
\begin{eqnarray}
  Y[E_j] ^a_{(0)\,\mathbf{k}} & = &\tfrac{1}{\kappa}\,
  \nabla^a Y{[ E_j]}_{\textbf{\text{k}}},
  \label{e:VectorHarmonicsE2-5Def0}\\
  Y[E_j]^ a_{(1)\,\mathbf{k}} & = &
  \tfrac{1}{\kappa k_3\sqrt{\kappa^2-k_3^2}}
    \left(\kappa^2\hat z^a\hat z^b-k_3^2g^{ab}\right)\nabla_bY[E_j]_\mathbf{k}
       ,\label{e:VectorHarmonicsE2-5Def1}\\
       Y[ E_j]^ a_{(2)\,\mathbf{k}} & = &
       \tfrac{1}{\sqrt{\kappa^2- k_3^2}}\,
           \epsilon^{a b c}\, \hat z_b \nabla_c Y[ E_j]_\mathbf{k}
    ,\label{e:VectorHarmonicsE2-5Def2}
\end{eqnarray}
so long as $\kappa^2\neq 0$, $k_3^2\neq 0$, and $\kappa^2\neq k_3^2$.
When $\kappa^2=0$, the vector harmonics must be spatially constant
vector fields.  Three linearly independent spatially constant vector
fields exist in the space $E_1$, so any orthonormal set can be used as
$\kappa=0$ harmonics in that space.  In the spaces $E_2,...,E_5$ the
only spatially constant unit vector field is $\hat z^a$, so it becomes
the only $\kappa=0$ vector field on those spaces.  When $k_3=0$ and
$\kappa\neq 0$ the vector harmonics are independent of $z$, i.e. they
depend only on $x$ and $y$.  In this case the class $A=0$ and $A=2$
harmonics are given by Eqs.~(\ref{e:VectorHarmonicsE2-5Def0}) and
(\ref{e:VectorHarmonicsE2-5Def2}), but the expression for the $A=1$
harmonic must be replaced by $ Y[ E_j]^ a_{(1)\,\mathbf{k}} = \hat z^a
Y[E_j]_\mathbf{k}$.  Finally if $\kappa^2=k_3^2$ and $\kappa\neq 0$
the vector harmonics are independent of $x$ and $y$,
i.e. the harmonics
depend only on $z$.  In this case there is only the single class of
vector harmonics given by Eq.~(\ref{e:VectorHarmonicsE2-5Def0}). It is
straightforward to show that the vector harmonics defined in
Eqs.~(\ref{e:VectorHarmonicsE2-5Def0})--(\ref{e:VectorHarmonicsE2-5Def2})
are eigenfunctions of the covariant Laplace operator with eigenvalue
$-\kappa^2$, satisfy orthonormality conditions that are the analogs of
those given in Eq.~(\ref{e:VectorOrthonormalityE1}), and satisfy
divergence identities that are the analogs of those given in
Eqs.~(\ref{e:VectorDivergenceE1Class0}) and
(\ref{e:VectorDivergenceE1Class2}).

The third (most general) approach to constructing vector harmonics on
these compact orientable flat spaces, $E_j$, is based on the fact that
these spaces are quotients, $\mathbb{E}^3/\Gamma$, of Euclidean space
$\mathbb{E}^3$ and an isometry group $\Gamma$.  Therefore the problem
of finding vector harmonics on the $E_j$ spaces is equivalent to
finding the $\Gamma$-invariant vector harmonics on $\mathbb{E}^3$.
This can be done using the method developed in
Ref.~\cite{riazuelo2004cosmic} to derive expressions for the scalar
harmonics.  The {\it Vector Action Lemma} and the {\it Vector
  Invariance Lemma} described in the Appendix to this paper provide
the tools needed to construct linear combinations of the
$\mathbb{E}^3$ vector harmonics that are invariant under all the
elements of the symmetry groups $\Gamma$.  The vector harmonics
obtained using this method on the spaces $E_1,...,E_5$ are more
complicated than those derived using the first two approaches.  So
here we report only the vector harmonics
$Y[E_6]^{a}_{(A)\text{\textbf{k}}}$ obtained in this way for the
otherwise intractable $E_6$ case:
\begin{eqnarray}
  Y[E_6]^{a}_{(A)\,\mathbf{k}} & = & \tfrac{1}{2} \left[ Y[E_1]^a_{(A)\,\mathbf{k}}
  + ( - 1)^{n_1 - n_2} M_1[E_6]^a{}_{b} Y[E_1]^b_{(A)\,\mathbf{k} \mathbf{M}_1[E_6]}
    \right. \nonumber \\
  && \left.\quad
  + ( - 1)^{n_2 - n_3} M_2[E_6]^a{}_{b}Y[E_1]^b_{(A)\,\mathbf{k} \mathbf{M}_2[E_6]}
+ ( - 1)^{n_3 - n_1} M_3[E_6]^a{}_{b} Y[ E_1]^ b_{(A)\,\mathbf{k}\mathbf{M}_3[E_6]}
  \right],
  \nonumber \\
  &&\text{for} \hspace{1em}
  \left( n_1, n_2 \in \mathbb{Z}^+, n_3 \in
  \mathbb{Z} \right) \hspace{1em} \text{or} \hspace{1em} \left( n_1 = 0, n_2, n_3
  \in \mathbb{Z}^+ \right)\hspace{1em}
    \nonumber \\
  &&
  \text{or} \hspace{1em}
  \left( n_2 = 0, n_1, n_3 \in
  \mathbb{Z}^+ \right),
  \label{e:VectorHarmonicsE6a}
  \\
  Y[E_6]^ a_{( A)\,(k_1, 0, 0) } & = & \tfrac{1}{\sqrt{2}}
  \left[ Y[ E_1]^a_{(A)\,(k_1,0, 0)} + Y[E_1]^ a_{(A)\,(- k_1, 0, 0 )}\right],
  \hspace{1em} \text{for} \hspace{1em}
  \left(n_1 \in 2 \mathbb{Z}^+, n_2 = n_3 = 0\right),
  \label{e:VectorHarmonicsE6b}
  \\
  Y[E_6]^ a_{( A)\,(0,k_2, 0) } & = & \tfrac{1}{\sqrt{2}}
  \left[ Y[ E_1]^a_{(A)\,(0,k_2,0)} + Y[E_1]^ a_{(A)\,(0,-k_2, 0)}\right],
  \hspace{1em} \text{for} \hspace{1em}
  \left(n_2 \in 2 \mathbb{Z}^+, n_1 = n_3 = 0\right),
  \label{e:VectorHarmonicsE6c}
  \\
  Y[E_6]^ a_{( A)\,(0,0,k_3) } & = & \tfrac{1}{\sqrt{2}}
  \left[ Y[ E_1]^a_{(A)\,(0,0,k_3)} + Y[E_1]^ a_{(A)\,(0,0,- k_3)}\right],
  \hspace{1em} \text{for} \hspace{1em}
  \left(n_3 \in 2 \mathbb{Z}^+, n_1 = n_2 = 0\right),
    \label{e:VectorHarmonicsE6d}
\end{eqnarray}
where $Y[E_1]^{a}_{(A)\text{\textbf{k}}}$ are the vector harmonics on
$E_1$ described above, and $A = 0, 1, 2$.  We note that the
expressions in
Eqs.~(\ref{e:VectorHarmonicsE6a})--(\ref{e:VectorHarmonicsE6d}) for
the $A=0$ case are equivalent to Eq.~(\ref{e:VectorHarmonic0Def}). It
is straightforward to show that the vector harmonics defined in
Eqs.~(\ref{e:VectorHarmonicsE6a})--(\ref{e:VectorHarmonicsE6d}) are
eigenfunctions of the covariant Laplace operator with eigenvalue
$-\kappa^2$, satisfy orthonormality conditions that are the analogs of
those given in Eq.~(\ref{e:VectorOrthonormalityE1}), and satisfy
divergence identities that are the analogs of those given in
Eqs.~(\ref{e:VectorDivergenceE1Class0}) and (\ref{e:VectorDivergenceE1Class2}).
The proofs of these properties use the fact that the matrices
$\mathbf{M}$ that define the symmetries of these spaces preserve the
inner product of vectors, e.g. for arbitrary vectors $\mathbf{u}$ and
$\mathbf{v}$, $\mathbf{u}\cdot\mathbf{v}=g(\mathbf{u},\mathbf{v})
=g(\mathbf{M}\cdot\mathbf{u},\mathbf{M}\cdot\mathbf{v})
=(\mathbf{M}\cdot\mathbf{u})\cdot(\mathbf{M}\cdot\mathbf{v})$.

Anti-symmetric tensor fields, $w_{a b} = - w_{b a}$, on orientable
three-manifolds are dual to the vector fields $v^a$. Thus,
for every $w_{a b}$ there exists a vector field $v^a$ so that $w_{a b}
= \epsilon_{a b c} v^c$. Therefore, any anti-symmetric tensor field
can be represented as a sum of vector harmonics.

\section{\label{sec5}Tensor harmonics}

This section constructs symmetric second-rank tensor harmonics on the
six flat compact orientable three-manifolds.  While these harmonics
are not unique and can be chosen in a variety of ways, our goal here
is to choose harmonics having three useful properties: First, the
tensor harmonics constructed here will be eigenfunctions of the
covariant Laplace operator:
\begin{equation}
  \nabla^c\nabla_c Y^{ab} = -\kappa^2 Y^{ab}.
  \label{e:TensorLaplaceEigenvalueEq}
\end{equation}
This ensures that these harmonics will form a complete basis for the
(square integrable) symmetric second-rank tensor fields on these
manifolds.  We use the notation $Y[E_j]^{ab}_{(A)\,\mathbf{k}}$ (for
$A=0,...,5$) to denote the six linearly independent classes of tensor
harmonics that satisfy Eq.~(\ref{e:TensorLaplaceEigenvalueEq}) on the
space $E_j$ (for $j=1,...,6$).  Second, the tensor harmonics
constructed here will satisfy nice orthonormality conditions,
i.e. $Y[E_j]^{ab}_{(A)\,\mathbf{k}}$ and $Y[E_j]^{ab}_{(A')\,\mathbf{k'}}$
will be orthogonal under the standard $L_2$ inner product unless
$\mathbf{k}=\mathbf{k'}$ and $A=A'$.  And third, the tensor harmonics
constructed here in classes $A=1,...,5$ will be trace free, and those
in classes $A=4$ and $A=5$ will be divergence free.

It is easy to construct four classes of tensor harmonics that satisfy
these properties on all the flat compact orientable three-manifolds.
These harmonics, which we call the class $A=0,...,3$ harmonics,
are defined as,
\begin{eqnarray}
  Y[E_j]^{ab}_{(0)\,\mathbf{k}}&=&\tfrac{1}{\sqrt{3}}g^{ab}Y[E_j]_{\mathbf{k}},
  \label{e:TensorHarmonic0Def}\\
  Y[E_j]^{ab}_{(1)\,\mathbf{k}}&=&\tfrac{1}{\sqrt{6}}
  \left(3\kappa^{-2}
  \nabla^a\nabla^bY[E_j]_{\mathbf{k}}
  +g^{ab}Y[E_j]_{\mathbf{k}}\right),
  \label{e:TensorHarmonic1Def}\\
  Y[E_j]^{ab}_{(2)\,\mathbf{k}}&=&\tfrac{1}{\kappa\sqrt{2}}
  \left(\nabla^aY[E_j]^b_{(1)\,\mathbf{k}}+\nabla^bY[E_j]^a_{(1)\,\mathbf{k}}\right),
  \label{e:TensorHarmonic2Def}\\
  Y[E_j]^{ab}_{(3)\,\mathbf{k}}&=&\tfrac{1}{\kappa\sqrt{2}}
  \left(\nabla^aY[E_j]^b_{(2)\,\mathbf{k}}+\nabla^bY[E_j]^a_{(2)\,\mathbf{k}}\right),
  \label{e:TensorHarmonic3Def}
\end{eqnarray}
where $\kappa^2$ is the corresponding eigenvalue of the covariant
scalar Laplace operator, Eq.~(\ref{e:ScalarEigenvalue}).  Since the
scalar harmonics $Y[E_j]_{\mathbf{k}}$ are smooth eigenfunctions of
the Laplace operator, their gradients are automatically smooth
eigenfunctions having the same eigenvalues on these flat manifolds.
The normalization factors in
Eqs.~(\ref{e:TensorHarmonic0Def})--(\ref{e:TensorHarmonic3Def}) are
chosen to ensure that these harmonics satisfy the nice orthonormality
conditions
\begin{equation}
  \int_{- L _1/ 2}^{L_1 / 2} \int_{- L_2 / 2}^{L _2/ 2} \int_{- L_3 / 2}^{L_3 / 2} g_{ac}
  g_{bd}Y{[E_j]}^{ab}_{(A)\,\mathbf{k}}
  Y{[ E_j]}^{cd*}_{(B)\,\mathbf{k}'} \,d x\, d y\, d z =
  \delta_{AB}\,\,\delta_{{n_1} {n_1}'}\,\delta_{{n_2} {n_2}'}\,\delta_{{n_3} {n_3}'}\,,
\end{equation}
for $A=0,...,3$ and $B=0,...,3$.  The traces of these tensor
harmonics are given by,
\begin{eqnarray}
  g_{ab}Y[E_j]^{ab}_{(0)\mathbf{k}}&=&\sqrt{3} Y[E_j]_\mathbf{k},
  \label{e:TraceTensor0}\\
  g_{ab}Y[E_j]^{ab}_{(A)\mathbf{k}}&=& 0,\hspace{1em}\text{for}\hspace{1em}
  A=1,...,3,
  \label{e:TraceTensor1-3}
\end{eqnarray}
while the divergences are given by
\begin{eqnarray}
  \nabla_aY{[E_j]}^{ab}_{(0)\text{\textbf{k}}}&=&\tfrac{\kappa}{\sqrt{3}}
    Y[E_j]^b_{(0)\,\mathbf{k}},
    \label{e:TensorHarmonicDivergence0}\\
  \nabla_aY{[E_j]}^{ab}_{(1)\text{\textbf{k}}}&=&-\tfrac{2\kappa}{\sqrt{6}}
    Y[E_j]^b_{(0)\,\mathbf{k}},
  \label{e:TensorHarmonicDivergence1}\\
  \nabla_aY[E_j]^{ab}_{(2)\,\mathbf{k}}&=&-\tfrac{\kappa}{\sqrt{2}} Y[E_j]^b_{(1)\,\mathbf{k}}\,,
    \label{e:TensorHarmonicDivergence2}\\
  \nabla_aY[E_j]^{ab}_{(3)\,\mathbf{k}}&=&-\tfrac{\kappa}{\sqrt{2}} Y[E_j]^b_{(2)\,\mathbf{k}}\,.
    \label{e:TensorHarmonicDivergence3}
\end{eqnarray}

The construction of the class $A=4$ and $A=5$ symmetric second-rank
tensor harmonics is less straightforward.  Our approach to finding
these harmonics is analogous to the methods described in
Sec.~\ref{sec4} for deriving the class $A=1$ and $A=2$ vector
harmonics.  We have developed three different approaches that produce
tensor harmonics satisfying the three useful properties listed above.
The first approach produces the simplest expressions for the tensor
harmonics, but this approach only works on the space $E_1$.  A
somewhat more general approach can be used to derive fairly simple
expressions for these tensor harmonics on the spaces $E_1,...,E_5$, but
not on $E_6$.  We have also developed an even more general approach
capable of constructing tensor harmonics in all the $E_j$ spaces, but
the resulting expressions produced in this way are quite complicated.
Here we report the results of this general approach only for the
otherwise intractable $E_6$ case.

The first approach to constructing the needed tensor harmonics is
based on the fact that any covariantly constant tensor field $c^{ab}$
is invariant under the pure translation symmetry group of the
three-torus, $E_1$.  Therefore any tensor field of the form
$c^{ab}Y[E_1]_\mathbf{k}$ is an eigenfunction of the covariant Laplace
operator, and therefore a candidate tensor harmonic.  The choices for
six linearly independent symmetric tensors $c^{ab}$ will define the
six classes of tensor harmonics for this case.  The class $A=0$ and
$A=1$ tensor harmonics defined in Eqs.~(\ref{e:TensorHarmonic0Def})
and (\ref{e:TensorHarmonic1Def}) for the space $E_1$ have this form:
$Y[E_1]^{ab}_{(0)(\mathbf{k}}=\tfrac{1}{\sqrt{3}}g^{ab}
Y[E_1]_{\mathbf{k}}$ and
$Y[E_1]^{ab}_{(0)(\mathbf{k}}=\tfrac{1}{\sqrt{6}}\left(g^{ab}
-3\hat k^a\hat k^b\right)Y[E_1]_{\mathbf{k}}$.  All that is needed to
complete this simple approach are choices for four additional constant
tensors to define the class $A=2,...,5$ harmonics.  These choices are
easy to construct using the set of orthonormal vectors $\hat k^a$,
$\hat \ell^a$, and $\hat m^a$ constructed in Sec.~\ref{sec4}:
\begin{eqnarray}
  Y[E_1] ^{a b}_{(0)\,\mathbf{k}}  & = & \tfrac{1}{\sqrt{3}}\left(
  \hat k^a \hat k^b+\hat\ell^a\hat \ell^b+\hat m^a\hat m^b\right)
  Y[E_1]_\mathbf{k},  \label{e:TensorHarmonic0E1}\\
  Y[E_1]^{ab}_{ ( 1)\,\mathbf{k}} & = & \tfrac{1}{\sqrt{6}} \left( \hat\ell^a \hat\ell^b
  + \hat m^a \hat m^b -2 \hat k^a \hat k^b\right) Y[E_1]_\mathbf{k},
  \label{e:TensorHarmonic1E1}\\
  Y[E_1]^{ab}_{(2)\,\mathbf{k}} & = & -\tfrac{1}{\sqrt{2}} \left(
  \hat k^a \hat\ell^b + \hat k^b \hat\ell^a\right) Y[ E_1]_\mathbf{k},
  \label{e:TensorHarmonic2E1}\\
  Y[E_1]^{ab}_{(3)\,\mathbf{k}} & = & -\tfrac{1}{\sqrt{2}} \left(
  \hat k^a \hat m^b + \hat k^b \hat m^a\right) Y[ E_1]_\mathbf{k},
  \label{e:TensorHarmonic3E1}\\
  Y[E_1]^{ab}_{(4)\,\mathbf{k}} & = & \tfrac{1}{\sqrt{2}} \left(
  \hat\ell^a \hat\ell^b - \hat m^a \hat m^b\right) Y[ E_1]_\mathbf{k},
  \label{e:TensorHarmonic4E1}\\
  Y[E_1]^{ab}_{(5)\,\mathbf{k}} & = & \tfrac{1}{\sqrt{2}} \left(
  \hat\ell^a \hat m^b + \hat\ell^b \hat m^a\right) Y[ E_1]_\mathbf{k},
  \label{e:TensorHarmonic5E1}
\end{eqnarray}
The class $A=0,...,3$ harmonics given in
Eqs.~(\ref{e:TensorHarmonic0E1})--(\ref{e:TensorHarmonic3E1}) are
equivalent to the expressions given in
Eqs.~(\ref{e:TensorHarmonic0Def})--(\ref{e:TensorHarmonic3Def}).
All these tensor harmonics defined in
Eqs.~(\ref{e:TensorHarmonic0E1})--(\ref{e:TensorHarmonic5E1}) satisfy
each of the useful properties listed above.  They are eigenfunctions
of the covariant Laplace operator with eigenvalue $-\kappa^2$, and
satisfy the following orthonormality conditions,
\begin{equation}
  \int_{- L _1/ 2}^{L_1 / 2} \int_{- L_2 / 2}^{L _2/ 2} \int_{- L_3 / 2}^{L_3 / 2} g_{ac}g_{bd}Y{[
  E_1]}^{ab}_{(A)\,\mathbf{k}} Y{[ E_1]}^{cd*}_{(B)\,\mathbf{k}'} \,d x\, d y\, d z =
  \delta_{AB}\,\,\delta_{{n_1} {n_1}'}\,\delta_{{n_2} {n_2}'}\,\delta_{{n_3} {n_3}'}\,.
  \label{e:TensorOrthonormalityE1}
\end{equation}
The traces of these tensor harmonics on $E_1$ are given by,
\begin{eqnarray}
  g_{ab} Y[E_1] ^{a b}_{(0)\,\mathbf{k}}&=&\sqrt{3}Y[E_1]_\mathbf{k},
  \label{e:TensorTrace0E1}\\
  g_{ab}Y[E_1] ^{a b}_{(A)\,\mathbf{k}}&=& 0, \hspace{1em}\text{for}\hspace{1em}
  A=1,...,5,\label{e:TensorTrace1-5E1}
\end{eqnarray}
and their divergences are given by,
\begin{eqnarray}
  \nabla_aY{[E_1]}^{ab}_{(0)\,\mathbf{k}}&=&\tfrac{\kappa}{\sqrt{3}}
    Y[E_1]^b_{(0)\,\mathbf{k}}\,,
  \label{e:TensorDivergence0E1}\\
  \nabla_aY{[E_1]}^{ab}_{(1)\,\mathbf{k}}&=&-\tfrac{2\kappa}{\sqrt{6}}
    Y[E_1]^b_{(0)\,\mathbf{k}}\,,
  \label{e:TensorDivergence1E1}\\
  \nabla_aY[E_1]^{ab}_{(2)\,\mathbf{k}}&=&-\tfrac{\kappa}{\sqrt{2}}
  Y[E_1]^b_{(1)\,\mathbf{k}}\,,
    \label{e:TensorDivergence2E1}\\
    \nabla_aY[E_1]^{ab}_{(3)\,\mathbf{k}}&=&-\tfrac{\kappa}{\sqrt{2}}
    Y[E_1]^b_{(2)\,\mathbf{k}}\,,
    \label{e:TensorDivergence3E1}\\
    \nabla_aY[E_1]^{ab}_{(4)\,\mathbf{k}}&=&\nabla_aY[E_1]^{ab}_{(5)\,\mathbf{k}}=0.
    \label{e:TensorDivergence4-5E1}
\end{eqnarray}

The second approach to constructing tensor harmonics uses the fact
that the vector $\hat z^a$, the metric $g^{ab}$ and the tensor $\hat
z^a\hat z^b$ (where $\hat z^a$ is the unit vector along the $z$-axis)
are the only covariantly constant vector and tensor fields invariant
under all the symmetry groups of the manifolds $E_1,...,E_5$.
Therefore tensors like $g^{ab} Y[E_j]_\mathbf{k}$, $\hat z^a\hat
z^bY[E_j]_\mathbf{k}$, and $\hat z^a Y[E_j]^b_\mathbf{k}+\hat z^b
Y[E_j]^a_\mathbf{k}$ (for $j=1,...,5$) are eigenfunctions of the
covariant Laplace operator that can be used in the construction of the
tensor harmonics on these spaces.  We note that $\hat
z^a\nabla_aY[E_j]_\mathbf{k}=i\,k_3 Y[E_j]_\mathbf{k}$ and $\hat
z^a\nabla_aY[E_j]^b_{(A)\,\mathbf{k}}=i\,k_3
Y[E_j]^b_{(A)\,\mathbf{k}}$ in all of these cases.  The tensor
harmonics for these manifolds can therefore be taken to be
\begin{eqnarray}
 Y[E_j]^{ab}_{(0)\,\mathbf{k}}&=&\tfrac{1}{\sqrt{3}}g^{ab}Y[E_j]_{\mathbf{k}},
  \label{e:TensorHarmonic0}\\
  Y[E_j]^{ab}_{(1)\,\mathbf{k}}&=&\tfrac{1}{\kappa^{2}\sqrt{6}}
  \left(3
  \nabla^a\nabla^bY[E_j]_{\mathbf{k}}
  +g^{ab}\kappa^{2}Y[E_j]_{\mathbf{k}}\right),
  \label{e:TensorHarmonic1}\\
  Y[E_j]^{ab}_{(2)\,\mathbf{k}}&=&\tfrac{1}{\kappa\sqrt{2}}
  \left(\nabla^aY[E_j]^b_{(1)\,\mathbf{k}}
  +\nabla^bY[E_j]^a_{(1)\,\mathbf{k}}\right),
  \label{e:TensorHarmonic2}\\
  Y[E_j]^{ab}_{(3)\,\mathbf{k}}&=&\tfrac{1}{\kappa\sqrt{2}}
  \left(\nabla^aY[E_j]^b_{(2)\,\mathbf{k}}
  +\nabla^bY[E_j]^a_{(2)\,\mathbf{k}}\right),
  \label{e:TensorHarmonic3}\\
  Y[E_j]^{ab}_{(4)\,\mathbf{k}} & = & \tfrac{1}{2\kappa\, k_3^2
    \sqrt{2(\kappa^2 - k_3^2)}}
    \left[ k_3\left(\kappa^2+k_3^2\right)
    \biggl( \nabla^a Y[E_j]^b_{(1)\,\mathbf{k}}
  + \nabla^b Y[E_j]^a_{(1)\,\mathbf{k}} \right)\nonumber\\
  &&\qquad\qquad\qquad\left.+i\kappa^2\left(\kappa^2-3k_3^2\right)
  \left(\hat z^a Y[E_j]^b_{(1)\,\mathbf{k}}
  +\hat z^b Y[E_j]^a_{(1)\,\mathbf{k}}\right)\right.\nonumber\\
  &&\qquad\qquad\qquad\left.
  +2\kappa\sqrt{\kappa^2-k_3^2}
  \left(\kappa^2\hat z^a\hat z^b-
   k_3^2g^{ab}\right)Y[E_j]_\mathbf{k}\right],
    \label{e:TensorHarmonic4}\\
  Y[E_j]^{ab}_{(5)\,\mathbf{k}} & = & \tfrac{1}{\kappa
    \sqrt{2 \left( \kappa^2 - k_3^2\right)}} \left[ k_3
    \left( \nabla^a Y[E_j]^b_{(2)\,\mathbf{k}}
  + \nabla^b Y[E_j]^a_{(2)\,\mathbf{k}} \right)
  - i\kappa^2 \left(\hat z^a Y[E_j]^b_{(2)\,\mathbf{k}}
  +\hat z^b Y[E_j]^a_{(2)\,\mathbf{k}}\right) \right]\,, \quad \quad
    \label{e:TensorHarmonic5}
\end{eqnarray}
so long as the vector harmonics $Y[E_j]^a_{(1)\,\mathbf{k}}$ and
$Y[E_j]^a_{(2)\,\mathbf{k}}$ are well defined, and so long as
$\kappa^2\neq 0$, $k_3^2\neq 0$, and $\kappa^2\neq k_3^2$.  When
$\kappa^2=0$, the tensor harmonics must be spatially constant tensor
fields.  Six linearly independent spatially constant tensor fields
exist in the space $E_1$, so any orthonormal set can be used as
$\kappa=0$ harmonics in that space.  In the spaces $E_2,...,E_5$ the
only spatially constant tensor fields are $g^{ab}$ and $g^{ab}-3\hat
z^a\hat z^b$, so (suitably normalized) they are the only $\kappa=0$
tensor harmonics on those spaces. When $k_3=0$ and $\kappa\neq 0$ the
tensor harmonics are independent of $z$, i.e. they depend only on $x$
and $y$. In this case the expressions for the vector harmonics are
given in Sec.~\ref{sec4} and the expressions for the $A=0,...,3$ and
$5$ tensor harmonics are given in
Eqs.~(\ref{e:TensorHarmonic0})--(\ref{e:TensorHarmonic3}) and
(\ref{e:TensorHarmonic5}).  However the class $A=4$
tensor harmonics are not well defined in this case.  Finally if
$\kappa^2=k_3^2$ and $\kappa\neq 0$ the tensor harmonics are
independent of $x$ and $y$, i.e. the harmonics depend only on $z$.  In this case the class
$A=1$ and $A=2$ vector harmonics do not exist, so the only well
defined tensor harmonics are the class $A=0$ and $A=1$ harmonics. It is straightforward
to show that the tensor harmonics defined in
Eqs.~(\ref{e:TensorHarmonic0})--(\ref{e:TensorHarmonic5}) are
eigenfunctions of the covariant Laplace operator with eigenvalue
$-\kappa^2$, satisfy orthonormality conditions that are the analogs of
those given in Eq.~(\ref{e:TensorOrthonormalityE1}), satisfy trace
identities that are the analogs of those given in
Eqs.~(\ref{e:TensorTrace0E1}) and (\ref{e:TensorTrace1-5E1}), and
divergence identities that are the analogs of those given in
Eqs.~(\ref{e:TensorDivergence0E1})--(\ref{e:TensorDivergence4-5E1}).

The third (most general) approach to constructing tensor harmonics on
these compact orientable flat spaces, $E_j$, is based on the fact that
these spaces are quotients, $\mathbb{E}^3/\Gamma$, of Euclidean space
$\mathbb{E}^3$ and an isometry group $\Gamma$.  Therefore the problem
of finding tensor harmonics on the $E_j$ spaces is equivalent to
finding the $\Gamma$-invariant tensor harmonics on $\mathbb{E}^3$.
This can be done using the method developed in
Ref.~\cite{riazuelo2004cosmic} to derive expressions for the scalar
harmonics.  The {\it Tensor Action Lemma} and the {\it Tensor
  Invariance Lemma} described in the Appendix to this paper provide
the tools needed to construct linear combinations of the
$\mathbb{E}^3$ tensor harmonics that are invariant under all the
elements of the symmetry groups $\Gamma$.  The tensor harmonics
obtained using this method on the spaces $E_1,...,E_5$ are more
complicated than those derived using the first two approaches.  So
here we report only the tensor harmonics
$Y[E_6]^{ab}_{(A)\text{\textbf{k}}}$ with $A = 0,...,5$ obtained in
this way for the otherwise intractable $E_6$ case:
\begin{eqnarray}
  Y[E_6]^{ab}_{(A)\,\mathbf{k}} & = & \tfrac{1}{2} \biggl[ Y[ E_1]^{ab}_{(A)\,\mathbf{k}}
    + ( - 1)^{n_1 - n_2}  M_1[E_6]^a{}_c M_1[E_6]^b{}_d Y[E_1]^{cd}_{(A)\,\mathbf{k} M_1[E_6]}
\nonumber \\ && \quad
    +( - 1)^{n_2 - n_3}  M_2[E_6]^a{}_c M_2[E_6]^b{}_d Y[E_1]^{cd}_{(A)\,\mathbf{k} M_2[E_6]}
    \nonumber \\
    && \quad
    + ( - 1)^{n_3 - n_1}  M_3[E_6]^a{}_c M_3[E_6]^b{}_d Y[E_1]^{cd}_{(A)\,\mathbf{k} M_3[E_6]}
    \biggr]
   \nonumber\\
   &&
   \text{for}\hspace{1em}\left( n_1, n_2\in \mathbb{Z}^+, n_3 \in \mathbb{Z} \right)
   \hspace{1em}
   \text{or} \hspace{1em} \left( n_1 = 0, n_2, n_3 \in \mathbb{Z}^+ \right) \hspace{1em}
   \nonumber\\
   &&
   \text{or} \hspace{1em} \left( n_2 = 0, n_1, n_3 \in
  \mathbb{Z}^+ \right), \label{e:TensorHarmonicsE6a}\\
  Y[E_6]^{ab}_{(A)\,(k_1, 0, 0)} & = & \tfrac{1}{\sqrt{2}} \left[ Y[E_1]^{ab}_{(A)\,(k_1, 0, 0)}
    + Y[E_1]^{ab}_{(A)\,(- k_1, 0, 0)}\right] \hspace{1em}
  \text{for} \hspace{1em} \left(n_1 \in 2 \mathbb{Z}^+, n_2 = n_3 = 0\right),
  \label{e:TensorHarmonicsE6b}\\
  Y[E_6]^{ab}_{(A)\,(0, k_2, 0)} & = & \tfrac{1}{\sqrt{2}} \left[ Y[E_1]^{ab}_{(A)\,(0, k_2, 0)}
    + Y[E_1]^{ab}_{(A)\,(0, - k_2, 0)}\right] \hspace{1em}
  \text{for} \hspace{1em} \left(n_2  \in 2 \mathbb{Z}^+, n_1 = n_3 = 0\right),
  \label{e:TensorHarmonicsE6c}\\
  Y[E_6]^{ab}_{(A)\,(0, 0, k_3)} & = & \tfrac{1}{\sqrt{2}} \left[ Y[E_1]^{ab}_{(A)\,(0, 0, k_3)}
    + Y[E_1]^{ab}_{(A)\,(0, 0, - k_3)}\right] \hspace{1em}
  \text{for} \hspace{1em} \left(n_3 \in 2 \mathbb{Z}^+, n_1 = n_2 = 0\right). \quad \quad
  \label{e:TensorHarmonicsE6d}
\end{eqnarray}
where $Y[E_1]^{ab}_{(A)\,\mathbf{k}}$ are the tensor harmonics on
$E_1$ described above. We note that
  the expressions in
  Eqs.~(\ref{e:TensorHarmonicsE6a})--(\ref{e:TensorHarmonicsE6d}) for
  the $A=0,...,3$ cases are equivalent to
  Eqs.~(\ref{e:TensorHarmonic0Def})--(\ref{e:TensorHarmonic3Def}). It
is straightforward to show that the tensor harmonics defined in
Eqs.~(\ref{e:TensorHarmonicsE6a})--(\ref{e:TensorHarmonicsE6d}) are
eigenfunctions of the covariant Laplace operator with eigenvalue
$-\kappa^2$, satisfy orthonormality conditions that are the analogs of
those given in Eq.~(\ref{e:TensorOrthonormalityE1}), satisfy the trace
conditions given in
Eqs.~(\ref{e:TensorTrace0E1})--(\ref{e:TensorTrace1-5E1}), and satisfy
divergence identities that are the analogs of those given in
Eqs.~(\ref{e:TensorDivergence0E1})--(\ref{e:TensorDivergence4-5E1}).
The proofs of these properties use the fact that the matrices
$M^a{}_b$ that define the symmetries of these spaces preserve the
structure of the metric: $g_{ab}=M^c{}_aM^d{}_bg_{cd}$ and $g^{ab}=M^a{}_cM^b{}_dg^{cd}$.

\section{\label{sec6}Summary and Discussion}

This paper introduces a uniform notation for the scalar, vector and
tensor harmonics on the flat compact orientable three-manifolds $E_1, E_2,
..., E_6$.  The $Y[E_j]_\mathbf{k}$ represent the scalar harmonics on
the space $E_j$ (for $j=1,...,6$) with parameters
$\mathbf{k}=k_a=(k_1,k_2,k_3)$.  To enforce the appropriate
periodicities, these parameters must be given by
$k_a=2\pi\left(\tfrac{n_1}{L_1},\tfrac{n_2}{L_2},\tfrac{n_3}{L_3}\right)$
in the spaces $E_1$, $E_2$, $E_3$ and $E_6$, and
$k_a=2\pi\left(-\tfrac{n_2}{L_1},
\tfrac{2n_1-n_2}{\sqrt{3}L_2},\tfrac{n_3}{L_3}\right)$ in the spaces
$E_4$ and $E_5$, where the $n_1$, $n_2$ and $n_3$ are integers and
$L_1$, $L_2$ and $L_3$ are the periodicity lengths in each dimension.
(The lengths $L_1$ and $L_2$ must also be equal in the spaces $E_3$,
$E_4$ and $E_5$ to preserve the symmetries.)  Explicit expressions for
these scalar harmonics, constructed originally in
Ref.~\cite{riazuelo2004cosmic}, are summarized in Sec.~\ref{sec3}.
The $Y[E_j]^a_{(A)\,\mathbf{k}}$ represent the three classes of vector
harmonics, with $A=0,1,2$, on the space $E_j$.  Explicit expressions
for these harmonics are constructed in Sec.~\ref{sec4}.  Finally the
$Y[E_j]^{ab}_{(A)\,\mathbf{k}}$ represent the six classes of symmetric
second-rank tensor harmonics, with $A=0,1,...,5$, on the space $E_j$.
Explicit expressions for these harmonics are constructed in
Sec.~\ref{sec5}.  All these harmonics satisfy a number of useful
properties:

The scalar harmonics on the six compact orientable flat
three-manifolds $E_j$ (for $j=1,2, ..., 6$) are eigenfunctions of the
covariant Laplace operator,
\begin{equation}
  \nabla^b \nabla_b Y[ E_j]_{\mathbf{k}} = - \kappa^2 Y[ E_j],
\end{equation}
and satisfy the orthonormality conditions
\begin{equation}
  \int_{- L _1/ 2}^{L_1 / 2} \int_{- L_2 / 2}^{L _2/ 2} \int_{- L_3 / 2}^{L_3 / 2} Y{[
  E_j]}_{\text{\textbf{k}}} Y{[ E_j]}^*_{\text{\textbf{k}}'} \,d x\, d y\, d z =
  \delta_{{n_1} {n_1}'}\,\delta_{{n_2} {n_2}'}\,\delta_{{n_3} {n_3}'}\,.
\end{equation}
Therefore it is easy to represent any (square integrable) scalar field on these
manifolds in terms of these basis functions:
\begin{equation}
  f(\mathbf{x})=\sum_{n_1}\sum_{n_2}\sum_{n_3}f_{\,\mathbf{k}} Y[E_j],
\end{equation}
where the coefficients $f_{\,\mathbf{k}}$ are given by
\begin{equation}
  f_{\,\mathbf{k}}=\int_{- L _1/ 2}^{L_1 / 2} \int_{- L_2 / 2}^{L _2/ 2} \int_{- L_3 / 2}^{L_3 / 2}
  f(\mathbf{x})\,Y{[ E_j]}^*_{\text{\textbf{k}}} \,d x\, d y\, d z.
\end{equation}

The vector harmonics on the six compact orientable flat
three-manifolds $E_j$ (for $j=1,2, ..., 6$) are eigenfunctions of the
covariant Laplace operator,
\begin{equation}
  \nabla^b \nabla_b Y[ E_j]^a_{(A)\,\mathbf{k}} = - \kappa^2 Y[ E_j]^a_{(A)\,\mathbf{k}}\,,
\end{equation}
and satisfy the orthonormality conditions
\begin{equation}
  \int_{- L _1/ 2}^{L_1 / 2} \int_{- L_2 / 2}^{L _2/ 2} \int_{- L_3 / 2}^{L_3 / 2} g_{ab}
  Y[E_j]^a_{(A)\,\mathbf{k}} Y[ E_j]^{b*}_{(B)\,\mathbf{k}'} \,d x\, d y\, d z =
  \delta_{AB}\,\,\delta_{{n_1} {n_1}'}\,\delta_{{n_2} {n_2}'}\,\delta_{{n_3} {n_3}'}\,.
\end{equation}
Therefore it is easy to represent any (square integrable) vector field on these
manifolds in terms of these basis functions:
\begin{equation}
  v^a(\mathbf{x})=\sum_{n_1}\sum_{n_2}\sum_{n_3}\sum_A v_{(A)\,\mathbf{k}}
  Y[E_j]^a_{(A)\,\mathbf{k}},
\end{equation}
where the coefficients $v_{(A)\,\mathbf{k}}$ are given by
\begin{equation}
  v_{(A)\,\mathbf{k}}=\int_{- L _1/ 2}^{L_1 / 2} \int_{- L_2 / 2}^{L _2/ 2} \int_{- L_3 / 2}^{L_3 / 2}
  g_{ab}v^a(\mathbf{x})\,Y[ E_j]^{b*}_{(A)\,\mathbf{k}} \,d x\, d y\, d z.
\end{equation}
These vector harmonics also satisfy the divergence identities,
\begin{eqnarray}
  \nabla_aY{[E_j]}^a_{(0)\,\mathbf{k}}&=&-\kappa Y[E_j]_{\mathbf{k}},
  \label{e:VectorDivergenceClass0}\\
  \nabla_aY{[E_j]}^a_{(1)\,\mathbf{k}}&=&
  \nabla_aY{[E_j]}^a_{(2)\,\mathbf{k}}=0,\label{e:VectorDivergenceClass2}
\end{eqnarray}
on each of the six flat compact orientable three-manifolds.  The
vanishing divergences of the class $A=1$ and $A=2$ vector harmonics
make them useful for constructing representations of the
electromagnetic field.

The tensor harmonics on the six compact orientable flat
three-manifolds $E_j$ (for $j=1,2, ..., 6$) are
eigenfunctions of the covariant Laplace operator,
\begin{equation}
  \nabla^c\nabla_c Y[E_j]^{ab}_{(A)\,\mathbf{k}} = -\kappa^2 Y[E_j]^{ab}_{(A)\,\mathbf{k}},
  \label{e:TensorLaplaceEigenvalueEqEj}
\end{equation}
and satisfy the orthonormality conditions
\begin{equation}
  \int_{- L _1/ 2}^{L_1 / 2} \int_{- L_2 / 2}^{L _2/ 2} \int_{- L_3 / 2}^{L_3 / 2} g_{ac}
  g_{bd}Y{[E_j]}^{ab}_{(A)\,\mathbf{k}}
  Y{[ E_j]}^{cd*}_{(B)\,\mathbf{k}'} \,d x\, d y\, d z =
  \delta_{AB}\,\,\delta_{{n_1} {n_1}'}\,\delta_{{n_2} {n_2}'}\,\delta_{{n_3} {n_3}'}\,,
\end{equation}
for $A=0,...,5$ and $B=0,...,5$.  Therefore it is easy to represent
any (square integrable) symmetric second-rank tensor field on these
manifolds in terms of these basis functions:
\begin{equation}
  t^{\,ab}(\mathbf{x})=\sum_{n_1}\sum_{n_2}\sum_{n_3}\sum_A t_{(A)\,\mathbf{k}}
  Y[E_j]^{ab}_{(A)\,\mathbf{k}},
\end{equation}
where the coefficients $t_{(A)\,\mathbf{k}}$ are given by
\begin{equation}
  t_{(A)\,\mathbf{k}}=\int_{- L _1/ 2}^{L_1 / 2} \int_{- L_2 / 2}^{L _2/ 2} \int_{- L_3 / 2}^{L_3 / 2}
  g_{ac}\,g_{bd}\,t^{\,ab}(\mathbf{x})\,Y[ E_j]^{cd*}_{(A)\,\mathbf{k}} \,d x\, d y\, d z.
\end{equation}
The traces of these tensor harmonics are given by,
\begin{eqnarray}
  g_{ab}Y[E_j]^{ab}_{(0)\,\mathbf{k}}&=&\sqrt{3} Y[E_j]_\mathbf{k},
  \label{e:TraceTensor0Sum}\\
  g_{ab}Y[E_j]^{ab}_{(A)\,\mathbf{k}}&=& 0,\hspace{1em}\text{for}\hspace{1em}
  A=1,...,5,
  \label{e:TraceTensor1-3Sum}
\end{eqnarray}
while the divergences are given by
\begin{eqnarray}
  \nabla_aY{[E_j]}^{ab}_{(0)\,\mathbf{k}}&=&\tfrac{\kappa}{\sqrt{3}}
    Y[E_j]^b_{(0)\,\mathbf{k}},
    \label{e:TensorHarmonicDivergence0Sum}\\
  \nabla_aY{[E_j]}^{ab}_{(1)\,\mathbf{k}}&=&-\tfrac{2\kappa}{\sqrt{6}}
    Y[E_j]^b_{(0)\,\mathbf{k}},
  \label{e:TensorHarmonicDivergence1Sum}\\
  \nabla_aY[E_j]^{ab}_{(2)\,\mathbf{k}}&=&-\tfrac{\kappa}{\sqrt{2}} Y[E_j]^b_{(1)\,\mathbf{k}}\,,
    \label{e:TensorHarmonicDivergence2Sum}\\
  \nabla_aY[E_j]^{ab}_{(3)\,\mathbf{k}}&=&-\tfrac{\kappa}{\sqrt{2}} Y[E_j]^b_{(2)\,\mathbf{k}}\,.
  \label{e:TensorHarmonicDivergence3Sum}\\
  \nabla_aY[E_j]^{ab}_{(4)\,\mathbf{k}}&=&\nabla_aY[E_j]^{ab}_{(5)\,\mathbf{k}}=0.
  \label{e:TensorHarmonicDivergence4-5Sum}
\end{eqnarray}
The vanishing traces and divergences of the class $A=4$ and $A=5$
tensor harmonics make them useful for constructing representations of
the dynamical gravitational wave degrees of freedom of the
gravitational field.

Cosmological observations provide measurements of
  physical properties, like the temperature of the cosmic microwave
  radiation, as functions on the two-sphere of our sky.  The
  components of the three-dimensional scalar, vector and tensor
  harmonics used in the construction of a cosmological model would
  need to be projected onto such a two-sphere to facilitate
  comparisons with the observations.  The analysis needed to do this
  projection has been carried out for the scalar harmonics on these
  flat three-dimensionsl manifolds in Ref.~\cite{riazuelo2004cosmic}.
  Analogous projections of the vector and tensor harmonics onto the
  appropriate spin-weighted spherical harmonic bases could also be carried
  out, but we defer that analysis to a future paper. 

\appendix
\section{Vector and Tensor Action and
  Invariance Lemmas\label{Appendix}}

All multi-connected three-dimensional flat spaces $E_j$ are quotients,
$\mathbb{E}^3/\Gamma$, of Euclidean space $\mathbb{E}^3$ by an
isometry group $\Gamma$.  The problem of finding the scalar, vector
and tensor harmonics on $\mathbb{E}^3/\Gamma$ is equivalent to the
problem of finding the $\Gamma$-invariant harmonics of $\mathbb{E}^3$.
Two technical lemmas were developed in Ref.~\cite{riazuelo2004cosmic}
to facilitate the construction of the scalar harmonics on these
spaces.  The first of these, the \textit{Action Lemma}, determines how
an element of the isometry group $\Gamma$ transforms one scalar
harmonic on $\mathbb{E}^3$ into another.  The second, the
\textit{Invariance Lemma}, constructs a harmonic that is invariant
under the repeated action of any element of the isometry group.
Together these lemmas were used in Ref.~\cite{riazuelo2004cosmic} to
derive the explicit expressions for the scalar harmonics summarized in
Sec.~\ref{sec3} of this paper.  Those lemmas for scalar harmonics are
generalized here to action and invariance lemmas for the vector and
tensor harmonics on these flat spaces.

Every isometry $\gamma\in\Gamma$ on these manifolds consists of a
reflection/rotation followed by a translation, i.e. they are
transformations of the form
$\mathbf{x}'=\mathbf{M}\cdot\mathbf{x}+\mathbf{T}$ where $\mathbf{M}$
is a unitary matrix and $\mathbf{T}$ is a vector, or in component
notation $x'^a=M^{\,a}{}_b x^b + T^a$.

\begin{lemma}
  \label{l:lemma1}
  (Vector Action Lemma)
The natural action of an isometry $\gamma \in \Gamma$ of Euclidean
space $\mathbb{E}^3$ takes a vector harmonic
$\textbf{\text{u}}
Y_{\mathbf{{k}}}(\mathbf{x})=\mathbf{u}\,e^{i\,\mathbf{k}\cdot\mathbf{x}}$
to another vector harmonic $\left( \mathbf{M}\cdot \mathbf{u} \right)
e^{i\, \mathbf{k} \cdot \mathbf{T}} Y_{\mathbf{k} \mathbf{M}} \left(
\mathbf{{x}} \right)$, where $\mathbf{u}$ is any constant vector.
\end{lemma}

\begin{proof}
  The action of the isometry $\gamma$ on $\mathbf{u}\,Y_\mathbf{k}$ is
  given by:
  \[ \mathbf{u}\, Y_{\mathbf{k}} \mapsto
    \gamma \left(\mathbf{u}\, e^{i\, \mathbf{k} \cdot \mathbf{x}} \right)
  = \gamma \left( \mathbf{u}  \right) \gamma \left( e^{i\,
    \mathbf{k} \cdot \mathbf{x}} \right)
  = \left( \mathbf{M}\cdot \mathbf{u} \right)
     e^{i\, \mathbf{k} \cdot \left( \mathbf{M}
    \cdot\mathbf{x} + \mathbf{T} \right)}
  = \left( \mathbf{M}\cdot \mathbf{u}  \right) e^{\text{i} \mathbf{k}
     \cdot \mathbf{T}} e^{i\, \mathbf{k} \mathbf{M}\cdot
     \mathbf{x}}
  = \left( \mathbf{M}\cdot \mathbf{u}  \right)
      e^{i\, \mathbf{k}\cdot\mathbf{T}}
     Y_{\mathbf{k} \mathbf{ M}} \left( \mathbf{x} \right), \]
  where we have used $\gamma \left( \mathbf{u}  \right) =
  \mathbf{u}' = \mathbf{M}\cdot \mathbf{u} $.
\end{proof}

\begin{lemma}
  \label{l:lemma2}
  (Vector Invariance Lemma) If $\gamma$ is an isometry of
  $\mathbb{E}^3$ with matrix part $\mathbf{M}$ and translational part
  $\mathbf{T}$, if $\mathbf{u}\,Y_{\mathbf{k}}$ is a vector
  harmonic on $\mathbb{E}^3$, and if $n$ is the smallest positive
  integer such that $\mathbf{k} = \mathbf{k}\mathbf{ M}^n$ (typically
  $n$ is simply the order of the matrix $\mathbf{M}$), then the action
  of $\gamma$
  \begin{enumeratenumeric}
  \item preserves the $n$-dimensional space of harmonics spanned by \\
    $\left\{ \mathbf{u}\, Y_{\mathbf{k}},
    \left(\mathbf{M}\cdot\mathbf{u}\right) Y_{\mathbf{k}\mathbf{M}},
      \cdots, \left(\mathbf{M}^{n - 1}\cdot \mathbf{u}\right)
      Y_{\mathbf{k} \mathbf{M}^{n - 1}} \right\}$ as a set, and

    \item leaves invariant the harmonic,
    $ a_0\, \mathbf{u}\, Y_{\mathbf{k}} + a_1 \left(
    \mathbf{M}\cdot \mathbf{u}  \right) Y_{\mathbf{k}\mathbf{M}}
    + \cdots + a_{n- 1} \left(  \mathbf{M}^{n - 1}\cdot \mathbf{u} \right)
       Y_{\mathbf{k}\mathbf{M}^{n - 1}}, $
    if and only if
    $ a_{j + 1} = e^{i\, \mathbf{k}\mathbf{M}^j\cdot \mathbf{T}} a_j$
       for each $j$ (mod  $n$).
  \end{enumeratenumeric}
\end{lemma}

\begin{proof}
  Both parts are immediate corollaries of Lemma 1. Specifically, the
  action of $\gamma$ takes the linear combination
\begin{equation}
  a_0\, \mathbf{u}\,
  Y_{\mathbf{k}} + a_1 \left( \mathbf{M}\cdot\mathbf{u} \right)
  Y_{\mathbf{k}\mathbf{M}} + \cdots + a_{n -2} \left( \mathbf{M}^{n -
    2}\cdot \mathbf{u} \right) Y_{\mathbf{k}\mathbf{M}^{n - 2}} + a_{n
    - 1} \left( \mathbf{M}^{n - 1}\cdot\mathbf{u} \right)
  Y_{\mathbf{k}\mathbf{M}^{n - 1}},
  \label{e:A1}
  \end{equation}
  into
  \begin{equation}
    a_0 \left( \mathbf{M}\cdot
  \mathbf{u} \right) e^{i \mathbf{k} \cdot \mathbf{T}}
  Y_{\mathbf{k}\mathbf{M}} + a_1 \left(\mathbf{M}^2\cdot
  \mathbf{u}\right) e^{i\, \mathbf{k} \cdot \mathbf{M} \mathbf{T}}
  Y_{\mathbf{k}\mathbf{M}^2} + \cdots + a_{n - 2} \left( \mathbf{M}^{n
    - 1}\cdot \mathbf{u} \right) e^{i\, \mathbf{k} \mathbf{M}^{n -
      2}\mathbf{T}} Y_{\mathbf{k}\mathbf{M}^{n - 1}} + a_{n - 1}
  \mathbf{u}\, e^{i\, \mathbf{k} \mathbf{M}^{n - 1}\cdot\mathbf{T}}
  Y_{\mathbf{k}}.
  \label{e:A2}
  \end{equation}
  Therefore the $n$-dimensional subspace spanned by $\left\{
  \mathbf{u}\, Y_{\mathbf{k}}, \left(\mathbf{M}\cdot\mathbf{u}\right)
  Y_{\mathbf{k}\mathbf{M}}, \cdots, \left(\mathbf{M}^{n -
    1}\cdot\mathbf{u} \right) Y_{\mathbf{k}\mathbf{M}^{n- 1}}
  \right\}$ is preserved as a set.  Comparing the coefficients of the
  expressions in Eqs.~(\ref{e:A1}) and (\ref{e:A2}), it follows that
  they are identical if only and if $a_{j + 1} = e^{i\,
    \mathbf{k}\mathbf{M}^j\cdot \mathbf{T}} a_j$ for each $j$ (mod
  $n$).
\end{proof}

\begin{lemma}
  \label{l:lemma3}
  (Tensor Action Lemma) The natural action of an isometry $\gamma$ of
  Euclidean space $\mathbb{E}^3$ takes a tensor harmonic
  $\mathbf{u}\otimes\mathbf{v}
  \,Y_{\mathbf{k}}=\mathbf{u}\otimes\mathbf{v}\,e^{i\,\mathbf{k}\cdot\mathbf{x}}$
  to another tensor harmonic $\left( \mathbf{M}\cdot \mathbf{u}
  \right)\otimes \left( \mathbf{M}\cdot \mathbf{v} \right) e^{i\, \mathbf{k}
    \cdot \mathbf{T}} Y_{\mathbf{k}\mathbf{M}} \left( \mathbf{x}
  \right)$, where $\mathbf{u}$ and $\mathbf{v}$ are arbitrary constant
  vectors.
\end{lemma}

\begin{proof}
  The action of the isometry $\gamma$ on
  $\mathbf{u}\otimes\mathbf{v}\,Y_\mathbf{k}$ is given by:
  \begin{align} \mathbf{u}\otimes\mathbf{v}\, Y_{\mathbf{k}} \mapsto \gamma
     \left( \mathbf{u}\otimes\mathbf{v}\, Y_{\mathbf{k}} \right) &=
     \gamma \left( \mathbf{u} \right)\otimes \gamma \left(
     \mathbf{v} \right)\, \gamma \left( Y_{\mathbf{k}} \right)
     = \left( \mathbf{M}\cdot \mathbf{u} \right)\otimes
     \left( \mathbf{M}\cdot \mathbf{v}
     \right) e^{i\, \mathbf{k}\cdot \left( \mathbf{M}\cdot \mathbf{x} +
     \mathbf{T} \right)} \notag \\
     &= \left( \mathbf{M}\cdot
     \mathbf{u} \right)\otimes \left( \mathbf{M}\cdot \mathbf{v} \right)
     e^{i\,\mathbf{k} \cdot \mathbf{T}} Y_{\mathbf{k}\mathbf{M}}
     \left( \mathbf{x} \right),  \end{align}
where we have used $\gamma \left( \mathbf{u}\otimes\mathbf{v}  \right) =
\mathbf{u}'\otimes\mathbf{v}' = \left(\mathbf{M}\cdot \mathbf{u}\right)
\otimes\left(\mathbf{M}\cdot \mathbf{v}\right)$.
\end{proof}

\begin{lemma}
  \label{l:lemma4}
  (Tensor Invariance Lemma) If $\gamma$ is an isometry of
  $\mathbb{E}^3$ with matrix part $\mathbf{M}$ and translational part
  $\mathbf{T}$, if $\mathbf{u}\otimes\mathbf{v}\,Y_{\mathbf{k}}$
  is a tensor harmonic on $\mathbb{E}^3$, and if $n$ is the smallest
  positive integer such that $\mathbf{k} = \mathbf{k}\mathbf{ M}^n$
  (typically $n$ is simply the order of the matrix $\mathbf{M}$), then
  the action of $\gamma$
  \begin{enumeratenumeric}
  \item preserves the $n$-dimensional space of harmonics spanned by\\
    $\left\{ \mathbf{u}\otimes\mathbf{v}\, Y_{\mathbf{k}},
    \left(\mathbf{M}\cdot\mathbf{u}\right)\otimes
    \left(\mathbf{M}\cdot\mathbf{v}\right) Y_{\mathbf{k}\mathbf{M}},
    \cdots, \left(\mathbf{M}^{n - 1}\cdot \mathbf{u}\right)
    \otimes\left(\mathbf{M}^{n - 1}\cdot \mathbf{v}\right)
      Y_{\mathbf{k} \mathbf{M}^{n - 1}} \right\}$ as a set, and

    \item leaves invariant the harmonic,
      $ a_0\, \mathbf{u}\otimes\mathbf{v}\, Y_{\mathbf{k}}
      + a_1 \left(\mathbf{M}\cdot \mathbf{u}  \right)\otimes
      \left(\mathbf{M}\cdot \mathbf{v}  \right) Y_{\mathbf{k}\mathbf{M}}
      + \cdots + a_{n- 1} \left(  \mathbf{M}^{n - 1}\cdot \mathbf{u} \right)
      \otimes\left(  \mathbf{M}^{n - 1}\cdot \mathbf{v} \right)
       Y_{\mathbf{k}\mathbf{M}^{n - 1}}, $
    if and only if
    $ a_{j + 1} = e^{i\, \mathbf{k}\mathbf{M}^j\cdot \mathbf{T}} a_j$
       for each $j$ (mod  $n$).
  \end{enumeratenumeric}
\end{lemma}

\begin{proof}
  Both parts are immediate corollaries of Lemma 3. Specifically, the
  action of $\gamma$ takes the linear combination
\begin{align}
  a_0\, \mathbf{u}\otimes\mathbf{v}\,
  Y_{\mathbf{k}} &+ a_1 \left( \mathbf{M}\cdot\mathbf{u} \right)
  \otimes\left( \mathbf{M}\cdot\mathbf{v} \right)
  Y_{\mathbf{k}\mathbf{M}} + \cdots
  + a_{n -2} \left( \mathbf{M}^{n -2}\cdot \mathbf{u} \right)
  \otimes\left( \mathbf{M}^{n -2}\cdot \mathbf{v} \right)
  Y_{\mathbf{k}\mathbf{M}^{n - 2}} \notag \\
  &+ a_{n- 1}
  \left( \mathbf{M}^{n - 1}\cdot\mathbf{u} \right)
  \otimes\left( \mathbf{M}^{n - 1}\cdot\mathbf{v} \right)
  Y_{\mathbf{k}\mathbf{M}^{n - 1}},
  \label{e:A3}
  \end{align}
  into
  \begin{eqnarray}
    &&a_0 \left( \mathbf{M}\cdot \mathbf{u} \right) \otimes\left(
    \mathbf{M}\cdot \mathbf{v} \right) e^{i \mathbf{k} \cdot
      \mathbf{T}} Y_{\mathbf{k}\mathbf{M}} + a_1
    \left(\mathbf{M}^2\cdot \mathbf{u}\right)
    \otimes\left(\mathbf{M}^2\cdot \mathbf{v}\right) e^{i\, \mathbf{k}
      \cdot \mathbf{M} \mathbf{T}} Y_{\mathbf{k}\mathbf{M}^2} + \cdots\nonumber\\
    &&\hspace{2cm}
    + a_{n - 2} \left( \mathbf{M}^{n- 1}\cdot \mathbf{u} \right)
    \otimes\left( \mathbf{M}^{n- 1}\cdot \mathbf{v} \right) e^{i\,
      \mathbf{k} \mathbf{M}^{n - 2}\mathbf{T}}
    Y_{\mathbf{k}\mathbf{M}^{n - 1}} + a_{n - 1} \mathbf{u}\otimes\mathbf{v}\, e^{i\,
      \mathbf{k} \mathbf{M}^{n - 1}\cdot\mathbf{T}} Y_{\mathbf{k}}.
  \label{e:A4}
  \end{eqnarray}
  Therefore the $n$-dimensional subspace spanned by \\$\left\{
  \mathbf{u}\otimes\mathbf{v}\,
  Y_{\mathbf{k}}, \left(\mathbf{M}\cdot\mathbf{u}\right)
  \otimes\left(\mathbf{M}\cdot\mathbf{v}\right)
  Y_{\mathbf{k}\mathbf{M}}, \cdots,
  \left(\mathbf{M}^{n - 1}\cdot\mathbf{u} \right)
  \otimes\left(\mathbf{M}^{n - 1}\cdot\mathbf{v} \right)
  Y_{\mathbf{k}\mathbf{M}^{n- 1}}
  \right\}$ is preserved as a set.  Comparing the coefficients of the
  expressions in Eqs.~(\ref{e:A3}) and (\ref{e:A4}), it follows that
  they are identical if only and if $a_{j + 1} = e^{i\,
    \mathbf{k}\mathbf{M}^j\cdot \mathbf{T}} a_j$ for each $j$ (mod
  $n$).
\end{proof}

\acknowledgments

We thank James Nester for helpful conversations
concerning this work. L.L. thanks the
  Morningside Center for Mathematics, Academy of Mathematics and
  Systems Science, Chinese Academy of Sciences, Beijing 100190, China
  for their hospitality during a visit in which a portion of this
  research was performed. This research was supported in part by the
National Natural Science Foundation of China grants 11503003 and
11633001, the Interdisciplinary Research Funds of Beijing Normal
University, the Strategic Priority Research Program of the Chinese
Academy of Sciences grant XDB23000000, and by the National Science
Foundation, USA grants PHY-1604244, DMS-1620366, and PHY-1912419.



\begin{thebibliography}{99}

\bibitem{WMAP_9yr_Results} G. Hinshaw, et. al, Astrophysical Journal
  Supplement Series, 208 (2) (2013) 19.

\bibitem{Feodoroff1885} E. Feodoroff,
  Russ. J. Crystallogr. Mineral. 21, 1 (1885).

\bibitem{Bieberbach1911} L. Bieberbach, Math. Ann. 70, 297 (1911); 72,
  400 (1912).

\bibitem{Novacki1934} W. Novacki, Comment. Math. Helv. 7, 81 (1934).

\bibitem{Geroch1968} R. Geroch, Journal of Mathematical Physics, 9 (1968) 1739.

\bibitem{Geroch1970} R. Geroch, Journal of Mathematical Physics, 11 (1970) 343.

\bibitem{riazuelo2004cosmic} A. Riazuelo, J. Weeks, J.-P. Uzan,
  R. Lehoucq, J.-P. Luminet, Physical Review D 69 (10) (2004) 103518.

\bibitem{sandberg1978tensor} V. D. Sandberg, Journal of Mathematical
  Physics 19 (12) (1978) 2441-2446.

\bibitem{jantzen1978tensor} R. T. Jantzen, Tensor harmonics on the
  3-sphere, Journal of Mathematical Physics 19 (5) (1978) 1163-1172.

\bibitem{rubin1984eigenvalues} M. A. Rubin, C. R. Ordonez, Journal of
  mathematical physics 25 (10) (1984) 2888-2894.

\bibitem{rubin1985symmetric} M. A. Rubin, C. R. Ordonez, Journal of
  mathematical physics 26 (1) (1985) 65-67.

\bibitem{lachieze2005laplacian} M. Lachieze-Rey, S. Caillerie,
  Classical and Quantum Gravity 22 (4) (2005) 695.

\bibitem{ben2016explicit} J. Ben Achour, E. Huguet, J. Queva,
  J. Renaud, Journal of Mathematical Physics 57 (2) (2016) 023504.

\bibitem{lindblom2017scalar} L. Lindblom, N. W. Taylor, F. Zhang,
  General Relativity and Gravitation 49 (11) (2017) 139.

\bibitem{clough2018robustness} K. Clough, R. Flauger, E. A. Lim,
  Journal of Cosmology and Astroparticle Physics 2018 (05) (2018) 065.

\bibitem{clough2018difficulty} K. Clough, J. C. Niemeyer, Classical
  and Quantum Gravity 35 (18) (2018) 187001.






\end{thebibliography}
\end{document}